\documentclass[draftclsnofoot, onecolumn, comsoc, 12pt]{IEEEtran} 
\ifCLASSINFOpdf
\else
   \usepackage[dvips]{graphicx}
\fi
\usepackage{url}
\usepackage{hhline}
\usepackage{graphicx}
\usepackage{multicol}
\usepackage{graphicx}
\usepackage{amssymb}
\usepackage{amsmath}

\usepackage{calc}
\usepackage{array}
\usepackage{optidef}
\usepackage{epstopdf}
\usepackage{mathtools}
\usepackage{mathrsfs}
\usepackage{bm}
\usepackage{cite}
\usepackage{boldline}
\usepackage{algorithm}
\usepackage{algpseudocode}

\DeclareMathOperator*{\minimize}{minimize}

\usepackage{amsthm}

\newtheorem{remark}{Remark}

\usepackage{booktabs,makecell,tabularx}

\newcolumntype{C}[1]{>{\centering\arraybackslash}p{#1}}
\newcolumntype{L}{>{\raggedright\arraybackslash}X}

\usepackage{siunitx}
\usepackage{etoolbox}
\newrobustcmd{\B}{\bfseries}

\usepackage{comment}
\usepackage{algorithm}
\usepackage{algpseudocode}
\usepackage{enumitem} 

\usepackage{colortbl}     
\usepackage{subfigure}
\usepackage{color, soul}
\usepackage{stfloats}
\usepackage{float}
\usepackage{multirow}
\usepackage{wrapfig}
\usepackage{booktabs}
\usepackage{slashbox}
\usepackage{optidef}
\definecolor{LightBlue}{rgb}{0.75,0.936,1.00}
\sethlcolor{LightBlue}
\definecolor{LightCyan}{rgb}{0.88,1,1}
\newtheorem{lemma}{Lemma}

\usepackage{amsmath}

\usepackage{xcolor}

\begin{document}
\bstctlcite{IEEEexample:BSTcontrol}
\title{{Rate-Matching Framework for RSMA-Enabled Multibeam LEO Satellite Communications}}

\author{Jaehyup Seong, Juha Park, Juhwan Lee, Jungwoo Lee,\\ Jung-Bin Kim,  Wonjae Shin, and H. Vincent Poor 
    \thanks{J. Seong, J. Park, and W. Shin are with 
    Korea University, Seoul 02841, South Korea 
    (email: {\texttt{\{jaehyup, juha, wjshin\}@korea.ac.kr}});
    J. Lee and J. Lee are with 
    Seoul National University, Seoul 08826, South Korea (email: {\texttt{\{sgsyk649, junglee\}@snu.ac.kr}});
    J.-B. Kim is with 
    ETRI, Daejeon 34129, South Korea (email: {\texttt{jbkim777@etri.re.kr}});
    H. V. Poor is with 
    Princeton University, Princeton, NJ 08544, USA (email: {\texttt{poor@princeton.edu}}).
    }    \thanks{A part of this work was presented in
part at the IEEE GLOBECOM Workshop, Kuala Lumpur, Malaysia, Dec. 2023 \cite{seong2023robust}.}}

\maketitle 
\begin{abstract}
With the goal of ubiquitous global connectivity, multibeam low Earth orbit (LEO) satellite communication (SATCOM) has attracted significant attention in recent years. 
The traffic demands of users are heterogeneous within the broad coverage of SATCOM due to different geological conditions and user distributions.
{Motivated by this, this paper proposes a novel rate-matching (RM) framework based on rate-splitting multiple access (RSMA) that minimizes the difference between the traffic demands and offered rates while simultaneously minimizing transmit power for power-hungry satellite payloads.} 
{Moreover, channel phase perturbations arising from channel estimation and feedback errors are considered to capture realistic multibeam LEO SATCOM scenarios.} 
To tackle the non-convexity of the RSMA-based RM problem under phase perturbations, we convert it into a tractable convex form via the successive convex approximation method and present an efficient algorithm to solve the RM problem. 
%
%
{Through the extensive numerical analysis across various traffic demand distribution and channel state information accuracy at LEO satellites, we demonstrate that RSMA flexibly allocates the power between common and private streams according to different traffic patterns across beams, thereby efficiently satisfying users' non-uniform traffic demands.}
{In particular, the use of common messages plays a vital role in overcoming the limited spatial dimension available at LEO satellites, enabling it to manage inter-/intra-beam interference effectively in the presence of phase perturbation.}
%
%
%
%

\end{abstract}

\IEEEpeerreviewmaketitle


\section{Introduction}
With the evolution of the principal source of traffic in wireless communications from mobile voice to mobile multimedia data due to the explosive development of high data rate wireless applications, the demand for massive connectivity and high throughput is continuously increasing \cite{zhong2018traffic}. 
Because of such requirements, broadband and multibeam satellite communications (SATCOM) have attracted considerable attention recently and are envisioned to play a vital role in 6G mobile communication systems.
The ubiquitous global connectivity of SATCOM enables it to serve not only urban areas but also rural and other remote areas, which cannot easily be served by conventional terrestrial base stations (BSs) \cite{perez2019signal}. 
In conventional multibeam SATCOM, the available bandwidth is divided into four sub-bands to reduce inter-beam interference efficiently, so-called four-color beam patterns \cite{perez2019signal}. 
Full frequency reuse is required to use the satellite spectrum efficiently while bringing about severe co-channel interference issues among the adjacent beams rather than four-color beam reuse.
To mitigate such inter- and intra-beam interference, accurate channel state information (CSI) needs to be available at both the transmitter (CSIT) and receiver (CSIR). 
However, obtaining accurate CSI in SATCOM is challenging because of the high end-to-end propagation delay, rapid mobility of satellites, and various other time-varying factors. 
{ For instance, low Earth orbit (LEO) satellites, orbiting at an altitude of 600 km with high velocities of 7.56 km/s, experience a significant round-trip delay (RTD) of 25.77 ms, substantially longer than on the few milliseconds typically observed in terrestrial networks \cite{10559954}.}
{ In typical satellite channel models, the magnitude of it is formulated with line-of-sight (LOS) and rain-attenuation components, determined by the path attenuation alone \cite{lin2021secrecy, lin2022refracting, lin2022slnr, gharanjik2015robust, zhang2019robust}.} Thus, the amplitude of the satellite channel is comparatively easy to estimate and varies slowly during the CSI feedback interval \cite{zhang2019robust}. 
{On the other hand, significant variations in the channel phase occur during the CSI feedback interval due to the high RTD, rapid satellite mobility, imperfect synchronization at the satellite, user movement, and other time-varying factors, resulting in severe phase perturbations \cite{vazquez2016precoding, gharanjik2015robust, zhang2019robust, wang2021resource}.}
Further, the phase perturbation can be caused even at the receiver side during the channel estimation process because of the deterioration of low-noise blocks and imperfect synchronization caused by the receiver's oscillator \cite{wang2021resource}. To provide reliable communication services across a broad coverage area of multibeam SATCOM, robust and advanced inter-/intra-beam interference management strategies are required.

Rate-splitting multiple access (RSMA) has recently emerged as a promising multiple access (MA) and interference management strategy that offers numerous benefits, including high spectral efficiency, robustness against imperfect CSI, network scalability, and flexibility \cite{clerckx2016rate, mao2018energy, mao2018rate, clerckx2019rate, park2023rate}. 
Such benefits of RSMA mainly come from splitting each user's message into common and private parts before being sent from the transmitter. Following the one-layer RSMA principle \cite{mao2018rate}, the common parts are combined and encoded into a single common stream for decoding by all users. Meanwhile, the private parts are individually encoded for decoding by corresponding users alone. The receivers first decode the common stream by considering the private streams as noise, then remove it from the received signal through the successive interference cancellation (SIC) technique. After SIC, the private streams are decoded, while the other private streams are treated as noise. It is worth pointing out that only a single SIC operation at the receiver is required in one-layer RSMA regardless of the scale of networks and user channel conditions.
One-layer RSMA offers comparable rate performance while simplifying both encoding and decoding complexities compared to the generalized RSMA that entails multiple common streams and necessitates multiple SIC operations at the receivers \cite{mao2018rate}. As such, by adjusting the portion between common and private parts, RSMA not only generalizes the existing MA techniques, including spatial division MA (SDMA), non-orthogonal MA (NOMA), and multicasting as special cases, but also provides spectral efficiency, robustness, scalability, and flexibility \cite{park2023rate}.
%

The utilization of RSMA in multibeam SATCOM is of increasing interest to effectively manage interference issues raised by imperfect CSI and provide reliable service to a wide range of networks with limited satellite radio resources \cite{yin2020rate, yin2020rate_J, yin2021ratephy, si2022rate, liu2023energy, khan2023rate, xu2023distributed, lin2021supporting, lin2020secure, yin2022rate, 10266774, kim2023distributed}.
{In both single- and multi-gateway SATCOM systems, RSMA was utilized to enhance max-min fairness (MMF) under imperfect CSIT \cite{yin2020rate, yin2020rate_J, yin2021ratephy, si2022rate}.}
{In \cite{liu2023energy}, an RSMA-based energy efficiency maximization problem was studied in multibeam SATCOM with imperfect CSIT.}
{To efficiently reuse the limited frequency band in coexistence networks of geostationary Earth orbit (GEO) and LEO satellites, an RMSA-based multi-layer SATCOM system was investigated to improve the sum-rate and MMF performance \cite{khan2023rate, xu2023distributed}.}
{RSMA was also employed in integrated satellite-aerial networks to maximize the sum-rate in \cite{lin2021supporting}.
Through \cite{lin2020secure, yin2022rate, 10266774, kim2023distributed}, it was proven that RSMA achieves better spectral and energy efficiencies compared to other MA techniques, such as SDMA and NOMA, in integrated satellite–terrestrial networks (ISTNs).}

%


{ 
As such, much of the existing works on RSMA in SATCOM \cite{yin2020rate, yin2020rate_J, yin2021ratephy, si2022rate, liu2023energy, khan2023rate, xu2023distributed, lin2021supporting, lin2020secure, yin2022rate, 10266774, kim2023distributed} focused on improving overall network quality by either maximizing the sum-rate or the minimum rate among users within the satellite coverage area. However, these approaches overlooked one of the key aspects of SATCOM: \textit{heterogeneous and time-varying nature of user traffic demands}.
It is important to emphasize that the traffic demands of users and beams within the satellite service area are potentially heterogeneous and time-varying, as satellites are capable of providing wide coverage across diverse regions, such as remote areas, urban centers, and suburban neighborhoods \cite{alberti2010system}. 
This heterogeneity in traffic demands can lead to inevitable mismatches between the demanded and offered rates in sum-rate maximization and MMF approaches \cite{yin2020rate, yin2020rate_J, yin2021ratephy, si2022rate, liu2023energy, khan2023rate, xu2023distributed, lin2021supporting, lin2020secure, yin2022rate, 10266774, kim2023distributed}. 

Without explicitly considering the asymmetry in traffic distribution, these mismatches can severely degrade system performance, either through under-provisioning the required rate (i.e., \textit{unmet rate}) or over-allocating resources beyond actual needs (i.e., \textit{unused rate}). These issues are particularly critical in the resource-constrained environment of SATCOM in which efficient utilization of power, frequency, and time resources is crucial. Satellites need to allocate their resources across various critical subsystems, including attitude/orbit control systems for maintaining satellite orientation, thermal regulation systems for managing temperature, and on-board payload operations such as remote sensing or Earth observation. Consequently, the resources available for communications are significantly restricted. To ensure the optimal use of the highly restricted resources for communications, the data rates offered to users must be carefully allocated according to their individual traffic demands. Therefore, designing a traffic-aware precoder is essential in SATCOM to dynamically adapt to the non-uniform traffic requirements of users and beams, while efficiently managing limited satellite radio resources.}

\subsection{Related Works}
{A number of studies considered the issue of traffic mismatches for SATCOM under heterogeneous traffic patterns among beams, including \cite{abdu2021flexible, ha2022geo, wang2020noma, lin2022multi, zheng2012generic, liu2019qos, 9769901, cui2023energy}.}
{The authors of \cite{abdu2021flexible, ha2022geo} discussed the transmission power minimization problem while guaranteeing the traffic demands with quality of service (QoS) constraints.}
{To be more specific, the authors of \cite{abdu2021flexible} focused on satisfying the QoS of each beam while jointly minimizing the carrier utilization and transmission power using SDMA.}
{In \cite{ha2022geo}, the authors investigated an optimal linear precoder and beam-hopping (BH) architecture to minimize the transmission power while satisfying the QoS of users.}
{ However, the power minimization problems with QoS constraints \cite{abdu2021flexible, ha2022geo} can result in infeasible solutions when the available power budget is insufficient or the channel condition is unfavorable for satisfying the traffic demands. This, in turn, poses fundamental challenges in reliably satisfying the traffic requirements of users in SATCOM, where power-hungry payloads are common.
Minimizing transmit power under QoS constraints risks destabilizing the system in SATCOM, which degrades the user experience. As such, SATCOM requires a more sophisticated approach that can directly address traffic mismatches and ensure reliable communication services, regardless of the available power budget or user channel conditions.} 

{To tackle this issue, the authors in \cite{wang2020noma} studied the flexible resource management problem based on NOMA to maximize the minimum value of the ratio of the offered capacity to the required traffic.}
{Similarly, in \cite{lin2022multi}, maximizing a minimum ratio of offered capacity to the traffic demand was investigated by considering an optimized BH architecture.} 
{In addition, the SDMA-based linear precoder design and NOMA-based beam power optimization problems were studied in \cite{zheng2012generic} and \cite{liu2019qos}, respectively, to minimize the difference between the traffic demands and offered rates.}
{The authors of \cite{9769901} 
focused on a joint optimization problem for resource allocation and beam scheduling to minimize the disparity between the traffic demands and offered rates based on NOMA.}
{Furthermore, the authors of \cite{cui2023energy} studied the joint {unmet rate} and transmission power minimization problem using minimum mean square error (MMSE)-based RSMA under perfect CSI conditions.} 

{\subsection{Motivations and Contributions}}

{ 
While substantial research efforts were dedicated to preventing the mismatches between required and offered rates directly \cite{wang2020noma, lin2022multi, zheng2012generic, liu2019qos, 9769901, cui2023energy}, previous studies face practical limitations that hinder their application in multibeam SATCOM with asymmetrically distributed demands for the following reasons.
\subsubsection*{\rm{1) Limited Spatial Dimension and Receiver Complexity}} In satellite networks, the number of users served by the satellite typically exceeds the number of transmit antennas, resulting in a user-overloaded scenario. However,
the previous works \cite{wang2020noma, lin2022multi, zheng2012generic, liu2019qos, 9769901} are based on SDMA or NOMA, both of which heavily depend on the network-load conditions.
In particular, as SDMA relies on the \textit{spatial degrees of freedom} provided by the multiple-input multiple-output technique, intra-/inter-beam interference cannot be sufficiently eliminated in user-overloaded scenarios, resulting in a high-level of residual interference. Moreover, the effectiveness of SDMA severely depends on the quality of CSIT, which is difficult to obtain perfectly at the satellite within a coherence time.
Also, NOMA requires increasingly complex receiver designs
due to SIC operations, as the number of users grows, and
without accurate CSIR, residual interference can cause an error propagation, reducing the effectiveness of interference control.
Thus, serving a large number of users with scarce degrees of
freedom in satellite radio resources has salient limitations.


\subsubsection*{\rm{2) Flexible RSMA Precoder Design According to Traffic Demands}}
In multibeam SATCOM, it is interesting to design a flexible and robust precoder according to heterogeneous traffic patterns among beams and users.
Although the authors of \cite{cui2023energy} utilized RSMA in overloaded scenarios and revealed its superiority over SDMA and NOMA for unmet rate reduction, the private precoding vectors are designed based on the conventional MMSE method that is not optimal for non-uniform traffic demands. There are still unsolved issues and room for improvement in designing a traffic-aware precoder to meet dynamic changes for multibeam SATCOM.


\subsubsection*{\rm{3) Phase Perturbation From Channel Estimation and Feedback Errors}}
The channel phase component in SATCOM varies rapidly because of the high RTD, movement of satellites and users, imperfect synchronization, and a variety of other time-varying factors. 
While the assumption that the magnitude of a channel is constant is acceptable to some extent since it is dominated by path attenuations, the rapidly time-varying phase component is not negligible \cite{vazquez2016precoding, gharanjik2015robust, zhang2019robust, wang2021resource}.
Therefore, the phase perturbations that can arise both in channel estimation and feedback procedures should be carefully considered. However, the effects of such phase perturbations were not considered in the literature \cite{wang2020noma, lin2022multi, zheng2012generic, liu2019qos, 9769901, cui2023energy}.

In summary, existing works \cite{wang2020noma, lin2022multi, zheng2012generic, liu2019qos, 9769901, cui2023energy} have not fully addressed the challenges in SATCOM with uneven traffic demands, including the scarcity of spatial degrees of freedom, the need for flexible precoder designs tailored to the traffic demands, and the impact of imperfect CSI, within a unified framework. While SDMA and NOMA can be effective under certain conditions, they struggle with interference management in overloaded environments and imperfect CSI, resulting in suboptimal performance. Moreover, the existing RSMA-based study relies on the conventional MMSE method, which is unsuitable for dynamically changing and asymmetric traffic demands across beams and users. Channel estimation and feedback errors, which were overlooked in previous studies, must also be carefully considered for practical SATCOM scenarios.

Motivated by this, we propose a novel rate-matching (RM) framework to comprehensively address these intertwined challenges in multibeam SATCOM. 
The proposed RM framework enhances system robustness through an adaptive traffic-aware RSMA precoder that effectively handles highly asymmetric traffic requirements while accounting for phase perturbations caused by channel estimation and feedback errors. The RSMA precoder is optimized to match the actual offered traffic to the non-uniform traffic demands, ensuring reliable communication services even under a limited power budget at the satellite. Moreover, our proposed optimization algorithm provides additional benefits in terms of transmit power efficiency.

In our previous study \cite{10304489}, we introduced a framework to prevent the unmet rates of multiple groups in multibeam GEO SATCOM by managing inter-beam interference using RSMA under perfect CSI. Continuing the same spirit in this work, we make improvements to the RSMA-based RM framework in addressing the non-uniform traffic demands of individual users in multibeam LEO SATCOM by mitigating both inter- and intra-beam interference under imperfect CSIT and CSIR caused by phase perturbations from channel estimation and feedback errors. Moreover, to further enhance resource utilization efficiency at the LEO satellite, we develop an optimization algorithm to simultaneously minimize transmit power and prevent both unmet/unused rates, efficiently utilizing the available power budget.
The main contributions are as follows:}

\begin{itemize}
\item We propose an RSMA-based RM framework that effectively satisfies individual user traffic demands by minimizing the difference between the demands and offered rates. {By matching the offered data rate to the data requirements with RSMA, the unused/unmet rates of users are effectively reduced in a multibeam LEO SATCOM environment with limited frequency and power resources. For more realistic LEO SATCOM scenarios, the channel phase perturbations caused by channel estimation and feedback errors (i.e., imperfect CSIT and CSIR) are also taken into account. To tackle the challenge of phase perturbations in precoder design for inter-/intra-beam interference management, we utilize statistical information on phase perturbations with relatively slower variations.}

\item To closely approximate the formulated non-convex problem into a tractable convex form, a successive convex approximation (SCA)-based algorithm is proposed.
In addition, we show that the proposed algorithm brings an additional benefit in terms of transmit power efficiency.
In other words, the proposed RM framework provides dual advantages, namely, the reduction of unused/unmet rates and the enhancement of transmit power efficiency.
\item 
{The performance of our proposed framework is validated based on key parameters from the 3rd Generation Partnership Project (3GPP) non-terrestrial networks (NTN) standards \cite{3gpp_ntn} to ensure realistic LEO SATCOM environments.} We show that the proposed RM framework outperforms the existing MMSE-based RSMA method \cite{cui2023energy} as well as the other MA techniques for {unused}/{unmet rate} reduction.
We numerically verify that the superiority of the proposed scheme stems from its key capabilities: i) flexible power allocation/precoding vector design according to traffic patterns among beams, ii) robustness against channel phase perturbations, and iii) network scalability. {This confirms RSMA as a potent MA technique to satisfy heterogeneous traffic demands in LEO SATCOM.}
\end{itemize}

\subsection{Notations}    
In the remaining sections of this paper, standard, lowercase boldface, and uppercase boldface letters indicate scalars, vectors, and matrices, respectively. 
The notation $\mathbf{X} \succcurlyeq 0$ indicates that matrix $\mathbf{X}$ is a positive semi-definite (PSD) matrix. 
Notations $(\cdot)^{\sf{T}}$, $(\cdot)^{\sf{H}}$, $\odot$, $\mathbb{E}[\cdot]$, $\Vert \cdot \Vert_{1}$, $\Vert \cdot \Vert_{2}$, and $\Vert \cdot \Vert_{\sf{F}}$ identify the transpose, hermitian, Hadamard product, expectation, L1-norm, L2-norm, and Frobenius norm, respectively. 
$\mathbf{1}$ denoted a vector of all 1's. 
Furthermore, $\mathbf{I}_{N}$ and $\mathbf{1}_{N}$ denote $N \times N$ identity matrix and matrix of all 1's, respectively. 

\section{System Model and Problem Formulation}

{ We consider a multibeam LEO SATCOM system, wherein a LEO satellite equipped with $N_{\sf{t}}$ antenna feeds (indexed by $\mathcal{N} \triangleq \{ 1,\cdots, N_{\sf{t}} \}$) serves $K$ users within the coverage area (indexed by $\mathcal{K} \triangleq \{ 1,\cdots, K \}$) equipped with a single antenna.}
{ We assume a single-feed-per-beam architecture \cite{perez2019signal} in which each antenna feed generates a distinct spot beam by employing a reflector antenna. By strategically positioning the feeds along the reflector antenna, which directs the signal from each feed, the satellite forms well-defined multiple spot beams aimed at specific coverage areas.}
{Herein, the LEO satellite illuminates various regions, and the users in each region have uneven traffic demands, as illustrated in Fig.~\ref{Fig1}.} 
{It is assumed that both the LEO satellite and users have perfect channel amplitude information but obtain imperfect CSI due to the channel phase perturbations caused by high RTD, rapid movement of LEO satellites, and other time-varying factors \cite{vazquez2016precoding, gharanjik2015robust, zhang2019robust, wang2021resource}.}
{Specifically, the LEO satellite has imperfect knowledge of CSI because of the phase perturbations caused by both channel estimation and feedback errors. In contrast, the users have imperfect knowledge of CSI because of the phase perturbation caused by the channel estimation error alone. Meanwhile, the LEO satellite is assumed to know exactly each user’s traffic demand.}

\subsection{Channel Model}
{In our system model, a channel between the LEO satellite and all users is characterized by matrix $\mathbf{H} = [\mathbf{h}_1, \cdots, \mathbf{h}_K] \in \mathbb{C}^{N_{\sf{t}} \times K}$, where vector $\mathbf{h}_k\in \mathbb{C}^{{N_{\sf{t}}}\times{1}}$ denotes the channel between the LEO satellite and the $k$-th user.}
The vector $\mathbf{h}_k$ is composed of the user terminal antenna gain, satellite beam radiation pattern, free-space loss, rain attenuation gain, and signal phase of the channel. 
Therefore, the $n_{\sf{t}}$-th element of $\mathbf{h}_k$, that is, the channel between the $n_{\sf{t}}$-th feed and $k$-th user, can be represented as follows \cite{yin2020rate_J, 7091022}:
\begin{align}\label{channel_Eq1}
h_{k,n_{\sf{t}}} = \frac{\sqrt{G_{{n_{\sf{t}}}, k}{G_{\sf R}}}}
{4\pi \frac{d_k}{\lambda} \sqrt{\kappa {T_{\sf sys}} {B}}} {\chi}_{n_{\sf{t}}, k}^{-1/2}e^{-j\phi_{n_{\sf{t}},k}}, 
\end{align}
where $G_{\sf R}$, $d_k$, $\lambda$, $\kappa$, $T_{\sf sys}$, and $B$ denote the user terminal antenna gain, distance between the satellite and $k$-th user, carrier wavelength, Boltzmann constant, receiving system noise temperature, and bandwidth, respectively.
$G_{{n_{\sf{t}}},k}$ is the beam gain from the $n_{\sf{t}}$-th feed to the $k$-th user, expressed using the $n$-th order first-kind Bessel function $J_{n}(\cdot)$ as  
\begin{align}\label{channelEq_2}
G_{{n_{\sf{t}}},k}=G_{\sf{max}}\bigg[\frac{J_{1}(\mu_{n_{\sf{t}},k})}{2\mu_{n_{\sf{t}},k}} + 36 \frac{J_{3}(\mu_{n_{\sf{t}},k})}{\mu_{n_{\sf{t}},k}^{3}}\bigg]^{2}, 
\end{align}
where $G_{\sf max}$ is the maximum beam gain observed at the beam center and $\mu_{n_{\sf{t}},k} = 2.07123 \, \mathrm{sin}(\theta_{{n_{\sf{t}}},{k}})/\mathrm{sin}(\theta_{3 \, \mathrm{dB}})$. 
In the equation $\mu_{n_{\sf{t}},k}$, the angles $\theta_{{n_{\sf{t}}},{k}}$ and $\theta_{3 \, \mathrm{dB}}$ denote the angle between the beam center and $k$-th user with respect to the $n_{\sf{t}}$-th feed and the 3 dB loss angle, respectively.
${\chi}_{n_{\sf{t}}, k}$ characterizes the effect of rain attenuation between the $n_{\sf{t}}$-th feed and $k$-th user the dB form of which ${\chi}_{n_{\sf{t}}, k}^{\mathrm{dB}}=10\log_{10}({\chi}_{n_{\sf{t}}, k})$ follows an independent log-normal distribution with mean $\mu$ and variance $\sigma^{2}$. $\phi_{n_{\sf{t}},k}$ is the signal phase component of the channel between the $n_{\sf{t}}$-th feed and $k$-th user that follows an independent uniform distribution from $0$ to $2\pi$.

\begin{figure}[!t]
\centering
 		\includegraphics[width=1\linewidth]{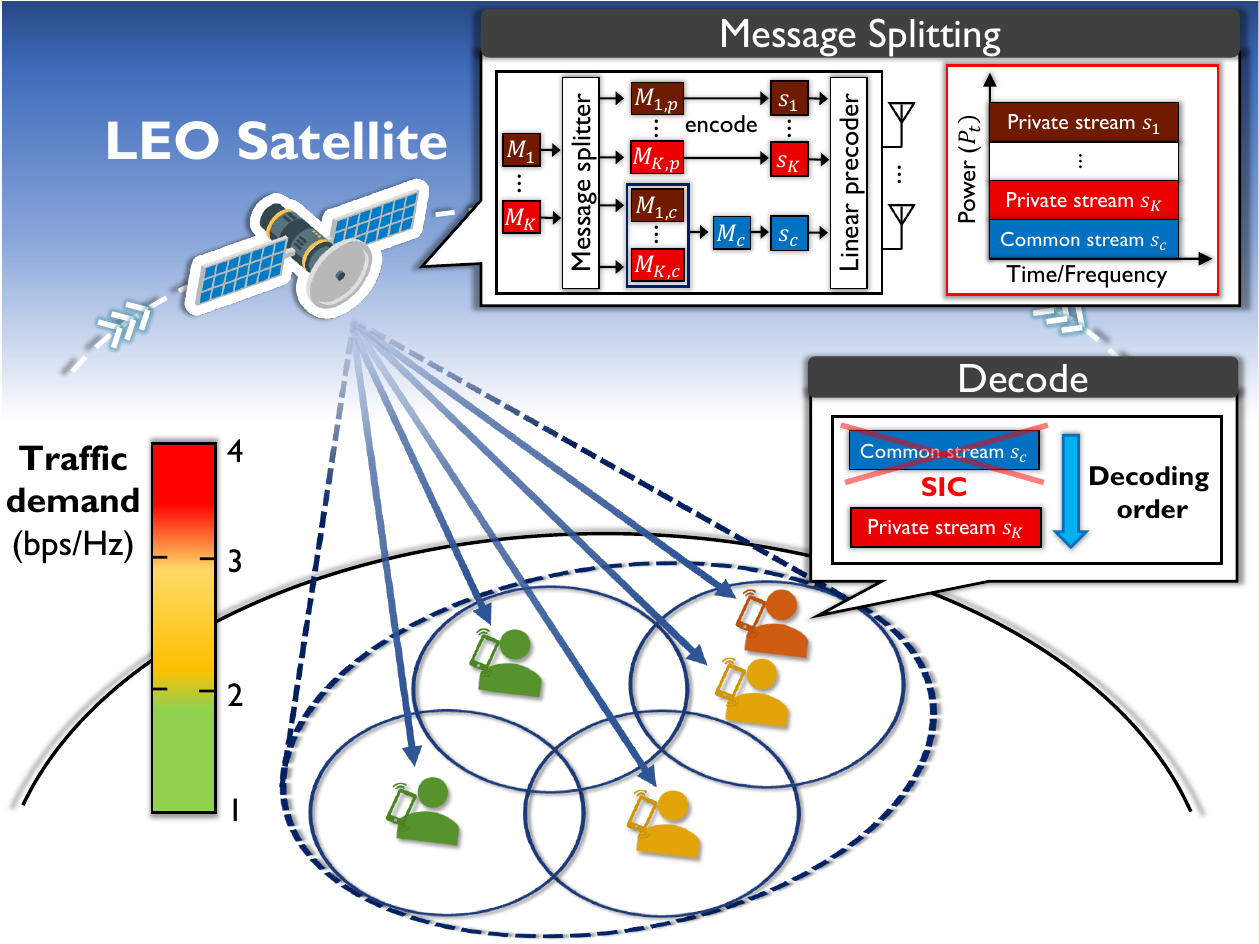}
 		\caption{{System model of the proposed RSMA scheme, where the LEO satellite serves users requiring different traffics in a wide range of service coverage.}}
    	\label{Fig1}
\end{figure}

\subsection{Imperfect Channel Model of the Satellite and Users}

As we assume that both satellite and users have imperfect channel phase information owing to the phase perturbations caused by the channel estimation and feedback processes, the satellite-to-$k$-th user channel vector $\mathbf{h}_k$ can be represented as
\begin{align}\label{channel_Eq3}
\mathbf{h}_k = \mathbf{\hat{h}}_k \odot {\mathbf{e}_{k}^{\sf{fb}}} \odot {\mathbf{e}_{k}^{\sf{ce}}} \in \mathbb{C}^{{N_{\sf{t}}}\times{1}}, 
\end{align}
{where $\mathbf{\hat{h}}_k\in \mathbb{C}^{{N_{\sf{t}}}\times{1}}$ and $\odot$ denote the partial CSI, which is known at the LEO satellite, and the Hadamard product, respectively.} In the equation (\ref{channel_Eq3}), ${\mathbf{e}_{k}^{\sf{fb}}} = [e^{j\theta_{k,1}^{\sf{fb}}}, \cdots, e^{j\theta_{k, N_{{\sf{t}}}}^{\sf{fb}}}]^{\sf{T}}\in \mathbb{C}^{{N_{\sf{t}}}\times{1}}$ and ${\mathbf{e}_{k}^{\sf{ce}}} = [e^{j\theta_{k,1}^{\sf{ce}}}, \cdots, e^{j\theta_{k, N_{{\sf{t}}}}^{\sf{ce}}}]^{\sf{T}}\in \mathbb{C}^{{N_{\sf{t}}}\times{1}}$, $\forall k \in \mathcal{K}$ 
denote the phase perturbations due to channel feedback and estimation errors, respectively. 
Since the satellite suffers from both channel estimation and feedback errors, the partial CSI known at the satellite is represented as 
$\hat{\mathbf{H}} = [\hat{\mathbf{h}}_1, \cdots, \hat{\mathbf{h}}_K]\in \mathbb{C}^{{N_{\sf{t}}}\times{K}}$. 
On the other hand, since the users suffer from the channel estimation error alone, the partial CSI known at the $k$-th user is represented as 
$\hat{\mathbf{h}}_k \odot \mathbf{e}_{k}^{\sf{fb}} \in \mathbb{C}^{{N_{\sf{t}}}\times{1}}$, $\forall k \in \mathcal{K}$. 
Herein, the phase perturbation vectors $\mathbf{\theta}_{k}^{\sf{fb}} = [\theta_{k,1}^{\sf{fb}}, \cdots, \theta_{k, N_{{\sf{t}}}}^{\sf{fb}}]^{\sf{T}}$ and $\mathbf{\theta}_{k}^{\sf{ce}} = [\theta_{k,1}^{\sf{ce}}, \cdots, \theta_{k, N_{{\sf{t}}}}^{\sf{ce}}]^{\sf{T}}$ are each assumed to be independent and identically distributed (i.i.d.) such that $\mathbf{\theta}_{k}^{\sf{fb}}\sim\mathcal{N}{(0,\delta_{\sf{fb}}^{2}\mathbf{I}_{N_{{\sf{t}}}})}$ and $\mathbf{\theta}_{k}^{\sf{ce}}\sim\mathcal{N}{(0,\delta_{\sf{ce}}^{2}\mathbf{I}_{N_{{\sf{t}}}})}$, $\forall k \in \mathcal{K}$, respectively. 
Under such imperfect CSIT and CSIR scenarios, we presume that statistical information on channel phase perturbations, which vary relatively slower, is available to be captured at the satellite.


\subsection{Rate-Splitting Multiple Access Based Signal Model}
Employing the rate-splitting strategy, the satellite first splits messages $M_{1}, \cdots, M_{K}$, intended for each user, into common messages and private messages, that is, $M_k \rightarrow \{M_{k,{\sf{c}}}, \,\, M_{k,{\sf{p}}} \}$, $\forall k \in \mathcal{K}$. 
Subsequently, the satellite combines all user's common messages $M_{1,{\sf{c}}}, \cdots, M_{K,{\sf{c}}}$ into one common message $M_{\sf{c}}$ and encodes it onto a common stream $s_{\sf{c}}$ using a codebook, which is known by all users. 
On the other hand, the private messages are encoded onto private streams $s_1, \cdots, s_K$, intended for each user, using codebooks known to the corresponding users only. 
Subsequently, using the precoding matrix $\mathbf{P}=[\mathbf{p}_{\sf{c}}, \mathbf{p}_1,\cdots,\mathbf{p}_K]\in \mathbb{C}^{{N_{\sf{t}}}\times{(K+1)}}$, consisting of the common precoding vector $\mathbf{p}_{\sf{c}}\in \mathbb{C}^{{N_{\sf{t}}}\times{1}}$ and private precoding vectors $\mathbf{p}_{k} \in \mathbb{C}^{{N_{{\sf{t}}}}\times{1}}$, $\forall k \in \mathcal{K}$, a stream vector $\mathbf{s}=[s_{\sf{c}}, s_1, \cdots, s_K]^{\sf{T}}\in \mathbb{C}^{{(K+1)}\times{1}}$ is linearly combined as 
\begin{align}
\label{TS_Eq}
\mathbf{x} = \mathbf{P}\mathbf{s} = \mathbf{p}_{\sf{c}} s_{\sf{c}} + \sum_{j=1}^{K}{\mathbf{p}_j s_j}\in \mathbb{C}^{{N_{\sf{t}}}\times{1}}, 
\end{align}
and then $\mathbf{x}$ is transmitted to the $k$-th user through $\mathbf{h}_k$. 
Thus, a signal received at the $k$-th user can be represented as follows:
\begin{align}
    \label{RS_Eq} 
    y_{k} & = \mathbf{h}_{k}^{\sf{H}}\mathbf{x} + n_{k}, \,\, \forall k \in \mathcal{K},
\end{align}
where $n_k$ is the additive white Gaussian noise (AWGN) that follows i.i.d. such that $n_k\sim\mathcal{CN}{(0,\sigma_{n}^{2})}$, $\forall k \in \mathcal{K}$.
At each user-$k$, the common stream $s_c$ is first decoded considering the private streams as noise and is removed from $y_k$ using SIC. 
Then, by treating the other private streams as noise, the dedicated private stream $s_k$ is decoded.

{Due to the phase perturbation by the estimation error, the received signal at the $k$-th user (\ref{RS_Eq}) is rewritten as
\begin{align}
    \label{RSC_Eq} 
    y_{k} & \, = \mathbf{h}_{k}^{\sf{H}}\mathbf{x} + n_{k} \nonumber \\  
    & \,= (\hat{\mathbf{h}}_{k} \odot {\mathbf{e}_{k}^{\sf{fb}}} \odot {\mathbf{e}_{k}^{\sf{ce}}})^{\sf{H}} \mathbf{x} + n_{k} \nonumber \\
    & \,= ((\hat{\mathbf{h}}_{k} \odot {\mathbf{e}_{k}^{\sf{fb}}}) + (\hat{\mathbf{h}}_{k} \odot {\mathbf{e}_{k}^{\sf{fb}}} \odot ({\mathbf{e}_{k}^{\sf{ce}}} - \mathbf{1})))^{\sf{H}} \mathbf{x} + n_{k} \nonumber \\
    & \overset{(a)}{=} \underbrace{(\hat{\mathbf{h}}_{k} \odot {\mathbf{e}_{k}^{\sf{fb}}})^{\sf{H}} \mathbf{p}_{c} s_{c}}_{\rm{desired \,\, common \,\, signal}}
    + \underbrace{(\hat{\mathbf{h}}_{k} \odot {\mathbf{e}_{k}^{\sf{fb}}} \odot ({\mathbf{e}_{k}^{\sf{ce}}} - \mathbf{1}))^{\sf{H}} \mathbf{p}_{\sf{c}} s_{\sf{c}}}_{\rm{self-interference \,\, signal}} \nonumber \\
    & \, + \underbrace{\sum_{j=1}^{K} (\hat{\mathbf{h}}_{k} \odot {\mathbf{e}_{k}^{\sf{fb}}} \odot {\mathbf{e}_{k}^{\sf{ce}}})^{\sf{H}} \mathbf{p}_{j} s_{j}}_{\rm{interference \,\, signals}} + \underbrace{n_{k}}_{\rm{noise}},
\end{align}
where step $(a)$ follows from (\ref{TS_Eq}), and the channel estimation error at the receiver induces the residual self-interference signal after SIC referred to as imperfect SIC.} 

To design a robust precoder in the presence of imperfect CSIT and CSIR, we first employ a concept of the generalized mutual information \cite{an2021rate,lee2022max, 9894281}. By treating the residual self-interference signal as an independent noise, we express the common signal-to-interference-plus-noise ratio (SINR) for the worst case under channel estimation error and represent it as $\gamma_{{\sf{c}},k}$.
Based on $\gamma_{{\sf{c}},k}$, we express the common rate in an ergodic form with respect to the feedback error as (\ref{RCEq1}) at the top of the next page in which $\mathcal{L} \triangleq \{{\sf c}, 1, \cdots, K\}$.
\begin{figure*}[!t]
\begin{align}
    \label{RCEq1}
    \small
    R_{{\sf{c}},k} &= \mathbb{E}_{\mathbf{\theta}_k^{\sf{fb}}} [\log_{2}{(1+\gamma_{{\sf{c}},k})}\vert\hat{\mathbf{h}}_k] \nonumber \\
    & = \mathbb{E}_{\mathbf{\theta}_k^{\sf{fb}}} \bigg[\log_{2} 
    \bigg(1 + \frac{\vert(\hat{\mathbf{h}}_{k} \odot \mathbf{e}_{k}^{\sf{fb}})^{\sf{H}} \mathbf{p}_{\sf{c}}\vert^{2}}{\sum\limits_{j \in \mathcal{L}}\mathbb{E}_{\mathbf{\theta}_k^{\sf{ce}}} [ \vert(\hat{\mathbf{h}}_{k} \odot \mathbf{e}_{k}^{\sf{fb}} \odot (\mathbf{e}_{k}^{\sf{ce}} - \mathbf{1}))^{\sf{H}} \mathbf{p}_{j}\vert^{2} ] + \sum\limits_{j=1}^{K} \vert (\hat{\mathbf{h}}_{k} \odot \mathbf{e}_{k}^{\sf{fb}})^{\sf{H}} \mathbf{p}_{j}\vert^{2} + \sigma_{n}^{2}}\bigg) \bigg] \nonumber \\
    & \, {\approx} \log_{2} 
    \bigg(1 + \frac{\mathbb{E}_{\mathbf{\theta}_k^{\sf{fb}}} [ \vert(\hat{\mathbf{h}}_{k} \odot \mathbf{e}_{k}^{\sf{fb}})^{\sf{H}} \mathbf{p}_{\sf{c}}\vert^{2} ] }{\sum\limits_{j \in \mathcal{L}}\mathbb{E}_{\mathbf{\theta}_k^{\sf{fb}}, \mathbf{\theta}_k^{\sf{ce}}} [ \vert(\hat{\mathbf{h}}_{k} \odot \mathbf{e}_{k}^{\sf{fb}} \odot (\mathbf{e}_{k}^{\sf{ce}} - \mathbf{1}))^{\sf{H}} \mathbf{p}_{j}\vert^{2} ] + \sum\limits_{j=1}^{K} \mathbb{E}_{\mathbf{\theta}_k^{\sf{fb}}} [\vert(\hat{\mathbf{h}}_{k} \odot \mathbf{e}_{k}^{\sf{fb}})^{\sf{H}} \mathbf{p}_{j}\vert^{2} ] + \sigma_{n}^{2}}\bigg) \nonumber \\
    &= \log_{2}
    \bigg(1+\frac{\mathbf{p}_{\sf{c}}^{\sf{H}}(\hat{\mathbf{h}}_{k} \hat{\mathbf{h}}_{k}^{\sf{H}} \odot
    [e^{-\delta_{\sf{fb}}^{2}}\mathbf{1}_{N_{{\sf{t}}}} + (1-e^{-\delta_{\sf{fb}}^{2}})\mathbf{I}_{N_{{\sf{t}}}}])\mathbf{p}_{\sf{c}}}
    {l_{{\sf{c}},k} + \sigma_{n}^{2}} \bigg).
\end{align}
\noindent\rule{\textwidth}{.05pt}
\end{figure*}
{ However, handling this ergodic form is challenging because there is typically no general closed-form expression.
To tackle this issue, we approximate it with the method described in \cite{6816003}, i.e., $\mathbb{E}[\log_2(1+X/Y)]\approx \log_2(1+\mathbb{E}[X]/\mathbb{E}[Y])$ in which $X$ and $Y$ are non-negative variables.} 
Thereafter, with the statistical information of ${\mathbf{e}_{k}^{\sf{fb}}}$ and ${\mathbf{e}_{k}^{\sf{ce}}}$, each term of the approximated equation in (\ref{RCEq1}) is reconstructed via the following procedure.
First, $\mathbb{E}_{\mathbf{\theta}_k^{\sf{fb}}}[\vert(\hat{\mathbf{h}}_{k} \odot {\mathbf{e}_{k}^{\sf{fb}}})^{\sf{H}} \mathbf{p}_{\sf{c}} \vert^{2}]$ is rewritten as 
\begin{align}
\label{Derv_1} 
& \mathbb{E}_{\mathbf{\theta}_k^{\sf{fb}}}[\vert(\hat{\mathbf{h}}_{k} \odot {\mathbf{e}_{k}^{\sf{fb}}})^{\sf{H}} \mathbf{p}_{\sf{c}} \vert^{2}] 
= \mathbf{p}_{\sf{c}}^{\sf{H}}(\hat{\mathbf{h}}_{k}\hat{\mathbf{h}}_{k}^{\sf{H}}\odot\mathbb{E}_{\mathbf{\theta}_k^{\sf{fb}}}[{\mathbf{e}_{k}^{\sf{fb}}} {\mathbf{e}_{k}^{\sf{fb}}}^{\sf{H}}])\mathbf{p}_{\sf{c}} \nonumber \\ 
& = \mathbf{p}_{\sf{c}}^{\sf{H}}(\hat{\mathbf{h}}_{k}\hat{\mathbf{h}}_{k}^{\sf{H}}\odot[e^{-\delta_{\sf{fb}}^{2}}\mathbf{1}_{N_{{\sf{t}}}} + (1-e^{-\delta_{\sf{fb}}^{2}})\mathbf{I}_{N_{{\sf{t}}}}])\mathbf{p}_{\sf{c}}.
\end{align}
Second, $\sum_{j=1}^{K}\mathbb{E}_{\mathbf{\theta}_k^{\sf{fb}}}[\vert (\mathbf{\hat{h}}_{k}^{\sf{H}}\odot{\mathbf{e}_{k}^{\sf{fb}}})^{\sf{H}} \mathbf{p}_{j} \vert^{2}]$ is rewritten as
\begin{align}
\label{Derv_2} 
&\sum_{j=1}^{K}\mathbb{E}_{\mathbf{\theta}_k^{\sf{fb}}}[\vert (\mathbf{\hat{h}}_{k}^{\sf{H}}\odot{\mathbf{e}_{k}^{\sf{fb}}})^{\sf{H}} \mathbf{p}_{j} \vert^{2}]  
= \sum_{j=1}^{K}\mathbf{p}_{j}^{\sf{H}}(\hat{\mathbf{h}}_{k}\hat{\mathbf{h}}_{k}^{\sf{H}}\odot\mathbb{E}_{\mathbf{\theta}_k^{\sf{fb}}}[{\mathbf{e}_{k}^{\sf{fb}}} {\mathbf{e}_{k}^{\sf{fb}}}^{\sf{H}}])\mathbf{p}_{j}  \nonumber \\
& = \sum_{j=1}^{K} \mathbf{p}_{j}^{\sf{H}}(\hat{\mathbf{h}}_{k}\hat{\mathbf{h}}_{k}^{\sf{H}}\odot[e^{-\delta_{\sf{fb}}^{2}}\mathbf{1}_{N_{{\sf{t}}}} + (1-e^{-\delta_{\sf{fb}}^{2}})\mathbf{I}_{N_{{\sf{t}}}}])\mathbf{p}_{j}.
\end{align}
Third, owing to the independence between $\mathbf{\theta}_{k}^{\sf{fb}}$ and $\mathbf{\theta}_{k}^{\sf{ce}}$, $\sum_{j \in \mathcal{L}}\mathbb{E}_{\mathbf{\theta}_k^{\sf{fb}}, \mathbf{\theta}_k^{\sf{ce}}}[\vert(\hat{\mathbf{h}}_{k} \odot {\mathbf{e}_{k}^{\sf{fb}}} \odot ({\mathbf{e}_{k}^{\sf{ce}}} - \mathbf{1}))^{\sf{H}} \mathbf{p}_{j}\vert^{2}]$ is rewritten as
\begin{align}
\label{Derv_3} 
&\sum_{j \in \mathcal{L}}\mathbb{E}_{\mathbf{\theta}_k^{\sf{fb}}, \mathbf{\theta}_k^{\sf{ce}}}[\vert(\hat{\mathbf{h}}_{k} \odot {\mathbf{e}_{k}^{\sf{fb}}} \odot ({\mathbf{e}_{k}^{\sf{ce}}} - \mathbf{1}))^{\sf{H}} \mathbf{p}_{j}\vert^{2}] = 
\nonumber \\
& \sum_{j \in \mathcal{L}} \mathbf{p}_{j}^{\sf{H}}(\hat{\mathbf{h}}_{k}\hat{\mathbf{h}}_{k}^{\sf{H}} \odot\mathbb{E}_{\mathbf{\theta}_k^{\sf{fb}}}[{\mathbf{e}_{k}^{\sf{fb}}} {\mathbf{e}_{k}^{\sf{fb}}}^{\sf{H}}] \odot\mathbb{E}_{\mathbf{\theta}_k^{\sf{ce}}}[({\mathbf{e}_{k}^{\sf{ce}}}-\mathbf{1}) ({\mathbf{e}_{k}^{\sf{ce}}}-\mathbf{1})^{\sf{H}}])\mathbf{p}_{j} 
\nonumber \\
& = \sum_{j\in \mathcal{L}}\mathbf{p}_{j}^{\sf{H}}(\hat{\mathbf{h}}_{k}\hat{\mathbf{h}}_{k}^{\sf{H}} \odot[e^{-\delta_{\sf{fb}}^{2}}\mathbf{1}_{N_{{\sf{t}}}} + (1-e^{-\delta_{\sf{fb}}^{2}})\mathbf{I}_{N_{{\sf{t}}}}] 
\nonumber \\
& \odot[(1-e^{-\frac{\delta_{\sf{ce}}^{2}}{2}})^{2}\mathbf{1}_{N_{{\sf{t}}}} + (1-e^{-\delta_{\sf{ce}}^{2}})\mathbf{I}_{N_{{\sf{t}}}}])\mathbf{p}_{j}.
\end{align}
In the abovementioned equations (\ref{Derv_1})-(\ref{Derv_3}), $\mathbb{E}_{\mathbf{\theta}_{k}^{\sf{fb}}}[{\mathbf{e}_{k}^{\sf{fb}}} {\mathbf{e}_{k}^{\sf{fb}}}^{\sf{H}}]$ and $\mathbb{E}_{\mathbf{\theta}_{k}^{\sf{ce}}}[({\mathbf{e}_{k}^{\sf{ce}}}-\mathbf{1}) ({\mathbf{e}_{k}^{\sf{ce}}}-\mathbf{1})^{\sf{H}}]$ can be represented as follows:
\begin{align}
\label{Derv_4}
&\mathbb{E}_{\mathbf{\theta}_{k}^{\sf{fb}}}[{\mathbf{e}_{k}^{\sf{fb}}} {\mathbf{e}_{k}^{\sf{fb}}}^{\sf{H}}]= e^{-\delta_{\sf{fb}}^{2}} \mathbf{1}_{N_{{\sf{t}}}} + (1-e^{-\delta_{\sf{fb}}^{2}})\mathbf{I}_{N_{{\sf{t}}}} \succcurlyeq 0, \\
\label{Derv_5}
&\mathbb{E}_{\mathbf{\theta}_{k}^{\sf{ce}}}[({\mathbf{e}_{k}^{\sf{ce}}}-\mathbf{1}) ({\mathbf{e}_{k}^{\sf{ce}}}-\mathbf{1})^{\sf{H}}] = 
\nonumber \\
&(1-e^{-\frac{\delta_{\sf{ce}}^{2}}{2}})^{2}\mathbf{1}_{N_{{\sf{t}}}} + (1-e^{-\delta_{\sf{ce}}^{2}})\mathbf{I}_{N_{{\sf{t}}}} \succcurlyeq 0.
\end{align}
{The proof of (\ref{Derv_4}) and (\ref{Derv_5}) is presented in Appendix A.}
By substituting (\ref{Derv_1})-(\ref{Derv_3}) into the approximated equation of (\ref{RCEq1}), $R_{{\sf{c}},k}$ can be represented as the final equation of (\ref{RCEq1}) in which $l_{{\sf{c}},k}$ is given by
\begin{align}\label{lCEq}
& l_{{\sf{c}},k} = \sum\limits_{j \in \mathcal{L}} \mathbf{p}_{j}^{\sf{H}}(\hat{\mathbf{h}}_{k} \hat{\mathbf{h}}_{k}^{\sf{H}}\odot[e^{-\delta_{\sf{fb}}^{2}}\mathbf{1}_{N_{{\sf{t}}}} + (1-e^{-\delta_{\sf{fb}}^{2}})\mathbf{I}_{N_{{\sf{t}}}}] 
\nonumber \\ 
& \odot [(1-e^{-\frac{\delta_{\sf{ce}}^{2}}{2}})^{2}\mathbf{1}_{N_{{\sf{t}}}} + (1-e^{-\delta_{\sf{ce}}^{2}})\mathbf{I}_{N_{{\sf{t}}}}]) \mathbf{p}_{j} 
\nonumber \\
& + \sum_{j=1}^{K} \mathbf{p}_{j}^{\sf{H}}(\hat{\mathbf{h}}_{k} \hat{\mathbf{h}}_{k}^{\sf{H}} \odot [e^{-\delta_{\sf{fb}}^{2}} \mathbf{1}_{N_{{\sf{t}}}} + (1-e^{-\delta_{\sf{fb}}^{2}})\mathbf{I}_{N_{{\sf{t}}}}])\mathbf{p}_{j}.
\end{align}

{To characterize the ergodic rate for the private stream, we decompose the $k$-th user signal
after the SIC procedure as
\begin{align}
    \label{RSP_Eq} y_{k,\sf{SIC}} & \, = y_k - \underbrace{(\hat{\mathbf{h}}_{k} \odot {\mathbf{e}_{k}^{\sf{fb}}})^{\sf{H}} \mathbf{p}_{\sf{c}} s_{\sf{c}}}_{\rm{desired \,\, common \,\, signal}} \nonumber \\
    & \overset{(b)}{=} \underbrace{(\hat{\mathbf{h}}_{k} \odot {\mathbf{e}_{k}^{\sf{fb}}})^{\sf{H}} \mathbf{p}_{k} s_{k}}_{\rm{desired \,\, private \,\, signal}} + \underbrace{(\hat{\mathbf{h}}_{k} \odot {\mathbf{e}_{k}^{\sf{fb}}} \odot ({\mathbf{e}_{k}^{\sf{ce}}} - \mathbf{1}))^{\sf{H}} \mathbf{p}_{k} s_{k}}_{\rm{self-interference\,\,signal}} \nonumber \\
    & \, + \underbrace{(\hat{\mathbf{h}}_{k} \odot {\mathbf{e}_{k}^{\sf{fb}}} \odot ({\mathbf{e}_{k}^{\sf{ce}}} - \mathbf{1}))^{\sf{H}} \mathbf{p}_{\sf{c}} s_{\sf{c}}}_{\rm{SIC \,\, error}} \nonumber \\
    & \, + \underbrace{\sum_{j=1 , j \neq k}^{K} (\hat{\mathbf{h}}_{k} \odot {\mathbf{e}_{k}^{\sf{fb}}} \odot {\mathbf{e}_{k}^{\sf{ce}}})^{\sf{H}} \mathbf{p}_{j} s_{j}}_{\rm{interference \,\, signals}} + \underbrace{n_{k}}_{\rm{noise}},
\end{align}
where step $(b)$ follows from (\ref{RSC_Eq}).}
Herein, both the self-interference signal and SIC error are caused by the channel estimation error at the receiver.
By using a similar technique to that used in the case of the common message, we express the private SINR for the worst case under channel estimation error as $\gamma_{{\sf{p}},k}$.
Then, by utilizing $\gamma_{{\sf{p}},k}$, we express the private rate in an ergodic form with respect to the feedback error as (\ref{RPEq1}) at the top of this page and approximate it based on the method described in \cite{6816003}.
\begin{figure*}[!t]
\begin{align}
    \label{RPEq1}
    \small
    R_{{\sf{p}},k} &= \mathbb{E}_{\mathbf{\theta}_k^{\sf{fb}}} [\log_{2}{(1+\gamma_{{\sf{p}},k})}\vert\hat{\mathbf{h}}_k] \nonumber \\
    &=\mathbb{E}_{\mathbf{\theta}_k^{\sf{fb}}} \bigg[\log_{2} 
    \bigg(1 + \frac{\vert(\hat{\mathbf{h}}_{k} \odot \mathbf{e}_{k}^{\sf{fb}})^{\sf{H}} \mathbf{p}_{k}\vert^{2}}{\sum\limits_{j \in \mathcal{L}}\mathbb{E}_{\mathbf{\theta}_k^{\sf{ce}}} [ \vert(\hat{\mathbf{h}}_{k} \odot \mathbf{e}_{k}^{\sf{fb}} \odot (\mathbf{e}_{k}^{\sf{ce}} - \mathbf{1}))^{\sf{H}} \mathbf{p}_{j}\vert^{2} ] + \sum\limits_{j=1 , j\neq k}^{K} \vert (\hat{\mathbf{h}}_{k} \odot \mathbf{e}_{k}^{\sf{fb}})^{\sf{H}} \mathbf{p}_{j}\vert^{2} + \sigma_{n}^{2}}\bigg) \bigg] \nonumber \\
    &\, {\approx} \log_{2} 
    \bigg(1 + \frac{\mathbb{E}_{\mathbf{\theta}_k^{\sf{fb}}} [ \vert(\hat{\mathbf{h}}_{k} \odot \mathbf{e}_{k}^{\sf{fb}})^{\sf{H}} \mathbf{p}_{k}\vert^{2} ] }{\sum\limits_{j \in \mathcal{L}}\mathbb{E}_{\mathbf{\theta}_k^{\sf{fb}}, \mathbf{\theta}_k^{\sf{ce}}} [ \vert(\hat{\mathbf{h}}_{k} \odot \mathbf{e}_{k}^{\sf{fb}} \odot (\mathbf{e}_{k}^{\sf{ce}} - \mathbf{1}))^{\sf{H}} \mathbf{p}_{j}\vert^{2} ] + \sum\limits_{j=1 , j\neq k}^{K} \mathbb{E}_{\mathbf{\theta}_k^{\sf{fb}}} [\vert(\hat{\mathbf{h}}_{k} \odot \mathbf{e}_{k}^{\sf{fb}})^{\sf{H}} \mathbf{p}_{j}\vert^{2} ] + \sigma_{n}^{2}}\bigg) \nonumber \\
    &= \log_{2}
    \bigg(1+\frac{\mathbf{p}_{k}^{\sf{H}}(\hat{\mathbf{h}}_{k} \hat{\mathbf{h}}_{k}^{\sf{H}} \odot
    [e^{-\delta_{\sf{fb}}^{2}}\mathbf{1}_{N_{{\sf{t}}}} + (1-e^{-\delta_{\sf{fb}}^{2}})\mathbf{I}_{N_{{\sf{t}}}}])\mathbf{p}_{k}}
    {l_{{\sf{p}},k} + \sigma_{n}^{2}} \bigg).
\end{align}
\noindent\rule{\textwidth}{.5pt}
\end{figure*} 
Subsequently, with the statistical information of ${\mathbf{e}_{k}^{\sf{fb}}}$ and ${\mathbf{e}_{k}^{\sf{ce}}}$, each term of the approximated equation in (\ref{RPEq1}) is reconstructed as the following procedure.
First, $\mathbb{E}_{\mathbf{\theta}_{k}^{\sf{fb}}}[\vert(\hat{\mathbf{h}}_{k} \odot {\mathbf{e}_{k}^{\sf{fb}}})^{\sf{H}} \mathbf{p}_{k} \vert^{2}]$ is rewritten as
\begin{align}
\label{Derv_6} 
& \mathbb{E}_{\mathbf{\theta}_{k}^{\sf{fb}}}[\vert(\hat{\mathbf{h}}_{k} \odot {\mathbf{e}_{k}^{\sf{fb}}})^{\sf{H}} \mathbf{p}_{k} \vert^{2}] = 
\nonumber \\
& \mathbf{p}_{k}^{\sf{H}}(\hat{\mathbf{h}}_{k}\hat{\mathbf{h}}_{k}^{\sf{H}}\odot [e^{-\delta_{\sf{fb}}^{2}}\mathbf{1}_{N_{{\sf{t}}}} + (1-e^{-\delta_{\sf{fb}}^{2}})\mathbf{I}_{N_{{\sf{t}}}}])\mathbf{p}_{k}.
\end{align}
Second, $\sum_{j=1, j \neq k}^{K} \mathbb{E}_{\mathbf{\theta}_{k}^{\sf{fb}}}[\vert (\mathbf{\hat{h}}_{k}^{\sf{H}}\odot{\mathbf{e}_{k}^{\sf{fb}}})^{\sf{H}} \mathbf{p}_{j} \vert^{2}]$ is rewritten as
\begin{align}
\label{Derv_7} 
& \sum_{j=1, j \neq k}^{K} \mathbb{E}_{\mathbf{\theta}_{k}^{\sf{fb}}}[\vert (\mathbf{\hat{h}}_{k}^{\sf{H}}\odot{\mathbf{e}_{k}^{\sf{fb}}})^{\sf{H}} \mathbf{p}_{j} \vert^{2}]  =
\nonumber \\
& \sum_{j=1, j \neq k}^{K} \mathbf{p}_{j}^{\sf{H}}(\hat{\mathbf{h}}_{k}\hat{\mathbf{h}}_{k}^{\sf{H}}\odot[e^{-\delta_{\sf{fb}}^{2}}\mathbf{1}_{N_{{\sf{t}}}} + (1-e^{-\delta_{\sf{fb}}^{2}})\mathbf{I}_{N_{{\sf{t}}}} ])\mathbf{p}_{j}.
\end{align}
Third, owing to the independence between $\mathbf{\theta}_{k}^{\sf{fb}}$ and $\mathbf{\theta}_{k}^{\sf{ce}}$, $\sum_{j\in \mathcal{L}} \mathbb{E}_{\mathbf{\theta}_{k}^{\sf{fb}}, \mathbf{\theta}_{k}^{\sf{ce}}}[\vert(\hat{\mathbf{h}}_{k} \odot {\mathbf{e}_{k}^{\sf{fb}}} \odot 
({\mathbf{e}_{k}^{\sf{ce}}} - \mathbf{1}))^{\sf{H}} \mathbf{p}_{j}\vert^{2}]$ is rewritten as
\begin{align}
\label{Derv_8} 
& \sum_{j\in \mathcal{L}} \mathbb{E}_{\mathbf{\theta}_{k}^{\sf{fb}}, \mathbf{\theta}_{k}^{\sf{ce}}}[\vert(\hat{\mathbf{h}}_{k} \odot {\mathbf{e}_{k}^{\sf{fb}}} \odot 
({\mathbf{e}_{k}^{\sf{ce}}} - \mathbf{1}))^{\sf{H}} \mathbf{p}_{j}\vert^{2}] 
\nonumber \\
& =\sum_{j\in \mathcal{L}} \mathbf{p}_{j}^{\sf{H}}(\hat{\mathbf{h}}_{k}\hat{\mathbf{h}}_{k}^{\sf{H}} 
\odot[e^{-\delta_{\sf{fb}}^{2}}\mathbf{1}_{N_{{\sf{t}}}} + (1-e^{-\delta_{\sf{fb}}^{2}})\mathbf{I}_{N_{{\sf{t}}}}] 
\nonumber \\
& \odot[(1-e^{-\frac{\delta_{\sf{ce}}^{2}}{2}})^{2}\mathbf{1}_{N_{{\sf{t}}}} + (1-e^{-\delta_{\sf{ce}}^{2}})\mathbf{I}_{N_{{\sf{t}}}}])\mathbf{p}_{j}.
\end{align}
By substituting (\ref{Derv_6})-(\ref{Derv_8}) into the approximated equation of (\ref{RPEq1}), $R_{{\sf{p}},k}$ can be represented as the final equation of (\ref{RPEq1}) in which $l_{{\sf{p}},k}$ is given by
\begin{align}\label{lPEq}
& l_{{\sf{p}},k} = \sum\limits_{j \in \mathcal{L}} \mathbf{p}_{j}^{\sf{H}}(\hat{\mathbf{h}}_{k} \hat{\mathbf{h}}_{k}^{\sf{H}}\odot[e^{-\delta_{\sf{fb}}^{2}}\mathbf{1}_{N_{{\sf{t}}}} + (1-e^{-\delta_{\sf{fb}}^{2}})\mathbf{I}_{N_{{\sf{t}}}}] 
\nonumber \\
& \odot [(1-e^{-\frac{\delta_{\sf{ce}}^{2}}{2}})^{2}\mathbf{1}_{N_{{\sf{t}}}} + (1-e^{-\delta_{\sf{ce}}^{2}})\mathbf{I}_{N_{{\sf{t}}}}])\mathbf{p}_{j} + \nonumber \\
& \sum_{j=1, j \neq k}^{K} \mathbf{p}_{j}^{\sf{H}}(\hat{\mathbf{h}}_{k} \hat{\mathbf{h}}_{k}^{\sf{H}} \odot [e^{-\delta_{\sf{fb}}^{2}} \mathbf{1}_{N_{{\sf{t}}}} + (1-e^{-\delta_{\sf{fb}}^{2}})\mathbf{I}_{N_{{\sf{t}}}}])\mathbf{p}_{j}.
\end{align}
 \begin{remark} 
 {\rm \textbf{(Ideal assumption of  perfect CSIR)}:
  Note that our derived rate expressions are general because they consider both the perfect CSIT and CSIR and the imperfect CSIT and perfect CSIR cases for the following reasons. In the imperfect CSIT and perfect CSIR case, because 
$\delta_{\sf{ce}}$ is $0^{\circ}$, $(1-e^{-\frac{\delta_{\sf{ce}}^{2}}{2}})^{2}\mathbf{1}_{N_{{\sf{t}}}} + (1-e^{-\delta_{\sf{ce}}^{2}})\mathbf{I}_{N_{{\sf{t}}}}$ becomes $\mathbf{0}_{N_{{\sf{t}}}}$. 
Therefore, (\ref{Derv_3}) and (\ref{Derv_8}) become $0$; thus, the derived terms (\ref{RCEq1}) and (\ref{RPEq1}) reduce to the rate expressions under the imperfect CSIT and perfect CSIR case presented in \cite{yin2022rate}. In the perfect CSIT and CSIR case, because 
$\delta_{\sf{fb}}$ is also $0^{\circ}$, $e^{-\delta_{\sf{fb}}^{2}}\mathbf{1}_{N_{{\sf{t}}}} + (1-e^{-\delta_{\sf{fb}}^{2}})\mathbf{I}_{N_{{\sf{t}}}}$ becomes $ \mathbf{1}_{N_{{\sf{t}}}}$. Therefore, (\ref{Derv_1}), (\ref{Derv_2}), (\ref{Derv_6}), and (\ref{Derv_7}) become $\vert \mathbf{h}_{k}^{\sf{H}} \mathbf{p}_{\sf{c}} \vert^{2}$, $\sum_{j=1}^{K} \vert \mathbf{h}_{k}^{\sf{H}} \mathbf{p}_{j} \vert^{2}$, $\vert \mathbf{h}_{k}^{\sf{H}} \mathbf{p}_{k} \vert^{2}$, and $\sum_{j=1, j \neq k}^{K} \vert \mathbf{h}_{k}^{\sf{H}} \mathbf{p}_{j} \vert^{2}$, respectively. Thus, the derived terms (\ref{RCEq1}) and (\ref{RPEq1}) reduce to the rate expressions under the perfect CSIT and CSIR case. }
\end{remark} 

\subsection{Problem Formulation}
Our work aims to design an RSMA precoder that matches the actual offered rates to the uneven traffic demands of users across the targeted coverage of the satellite with a limited power budget.
To this end, the optimization problem for the RSMA-based RM framework is formulated as follows: \footnote{ Although the objective function of the formulated problem is based on the L2-norm, it can be readily converted into an L1-norm-based formulation by simply modifying the objective function while retaining the same constraints and optimization process. The respective advantages of both approaches will be discussed in the numerical studies.}
\begin{align}
\nonumber
\mathscr{P}_1: \,\,\,\, 
\minimize_{\mathbf{P}, \mathbf{c}} \,\, \sum_{j=1}^{K} \vert R_{{\sf target}{,j}}-R_{j} \vert^{2}
\end{align}\setcounter{equation}{19}
\begin{subequations}\label{condition}
\begin{align}
\mathrm{s.t.}\,\,\,\,\,\,
\label{PF1CST1}
&R_{{\sf c},k} \geq \sum_{j=1}^{K}C_{j}, \,\, \forall k \in \mathcal{K},\\
\label{PF1CST2}
&C_k \geq 0, \,\, \forall k \in \mathcal{K}, \\
\label{PF1CST3}
&{(\mathbf{P}\mathbf{P}^{\sf{H}}})_{{n_{\sf{t}}},{n_{\sf{t}}}} \leq \frac{P_{\sf{t}}}{N_{\sf{t}}}, \,\, \forall {n_{\sf{t}}} \in \mathcal{N}.
\end{align}
\end{subequations}
In the formulated problem $\mathscr{P}_1$, $R_{{\sf target}{,k}}$ denotes the traffic demand of the $k$-th user, and $R_{k}$ denotes the total ergodic rate of the $k$-th user, which is represented as $R_{k}=R_{{\sf{p}},k}+C_{k}$. Herein, $C_{k}$ is the $k$-th user's portion from the common ergodic rate.
$\mathbf{P}=[\mathbf{p}_{\sf{c}}, \mathbf{p}_1,\cdots,\mathbf{p}_K]\in \mathbb{C}^{{N_{\sf{t}}}\times{(K+1)}}$ and $\mathbf{c}=[C_1,\cdots, C_K]^{\sf{T}}$ denote a precoding matrix and a vector consisting of common rate portions of the users, respectively. 
The constraint (\ref{PF1CST1}) enables $s_{\sf{c}}$ to be decodable by all users, and constraint (\ref{PF1CST2}) ensures the common rate portions to be non-negative values.
The constraint (\ref{PF1CST3}) represents the per-feed power constraint that reflects the inability to share energy resources among the dedicated high-power amplifiers (HPAs) of antenna feeds \cite{perez2019signal}. Meanwhile, $P_{\sf{t}}$ denotes the total power budget for the satellite.

Since $\mathscr{P}_1$ is a non-convex problem that is difficult to solve directly, we reformulate it into a tractable convex problem based on the SCA method in Section \uppercase\expandafter{\romannumeral3}.

\section{Proposed RSMA-Based Rate-Matching Framework}
To convert $\mathscr{P}_1$ into a convex problem, a slack variable $\alpha_{{\sf{p}},k}$ in regard to $R_{{\sf{p}},k}$ is introduced. Thus, $\mathscr{P}_1$ is reformulated as
\begin{align}
\nonumber
\mathscr{P}_2: \,\,\,\, 
\minimize_{\mathbf{P}, \mathbf{c}, \boldsymbol{\alpha}_{\sf{p}}} 
{\Vert {\mathbf{r}_{\sf target}- ( \mathbf{c} + \boldsymbol{\alpha}_{\sf{p}} )} \Vert}_{2}^{2}
\end{align}\setcounter{equation}{20}
\begin{subequations}
\begin{align}
\mathrm{s.t.}\,\,\,\,\,\,
\label{PF2CST1}
&R_{{\sf{p}},k} \geq \alpha_{{\sf{p}},k},   \,\, \alpha_{{\sf{p}},k} \geq 0, \,\, \forall k \in \mathcal{K}, \\
& \textrm{(\ref{PF1CST1})}, \,\, \textrm{(\ref{PF1CST2})}, \,\,  \textrm{(\ref{PF1CST3})}, \nonumber
\end{align}
\end{subequations}
where $\mathbf{r}_{\sf target}$ and $\boldsymbol{\alpha}_{\sf{p}}$ denote $\mathbf{r}_{\sf target}=[R_{{\sf target}, {1}},\cdots,R_{{\sf target}, {K}}]^{\sf{T}}$ and $\boldsymbol{\alpha}_{\sf{p}}=[\alpha_{{\sf{p}},1},\cdots,\alpha_{{\sf{p}},K}]^{\sf{T}}$, respectively. 
By doing so, we convert the problem into one that minimizes the difference between the traffic demands and the lower bound of $R_{k}$, i.e., $\alpha_{{\sf{p}},k} + C_{k}$, which turns out to be a convex function with respect to $[\mathbf{c}^{\sf{T}}, \boldsymbol{\alpha}_{\sf{p}}^{\sf{T}}]^{\sf{T}}$.
However, an optimal solution to the relaxed problem is always to transmit at the maximum allowable power levels at the satellite, even when plenty of offered rate is wasted and unused due to the lower bound constraint (\ref{PF2CST1}).
{ This poses a critical technical challenge in LEO SATCOM with uneven demands, where the offered rates need to be carefully allocated according to the individual traffic demands while efficiently utilizing the transmit power budget at LEO satellites, given the inherently power-hungry nature of their payloads.
To address this issue, we incorporate a power minimization term into the objective function, weighted by the regularization parameter $\eta$ ranging from 0 to 1 bps/Hz/W, as follows:}
\begin{align}
\nonumber
\mathscr{P}_3: \,\,\,\, 
\minimize_{\mathbf{P}, \mathbf{c}, \boldsymbol{\alpha}_{{\sf{p}}}} \,
\eta{\Vert {\mathbf{r}_{\sf target}- (\mathbf{c} + \boldsymbol{\alpha}_{\sf{p}})} \Vert}_{2}^{2} + (1-\eta) \Vert {\mathbf{P}} \Vert_{\sf{F}}^{2}
\end{align}\setcounter{equation}{21}
\begin{subequations}
\begin{align}
\mathrm{s.t.}\,\,\,\,\,\,
& \textrm{(\ref{PF1CST1})}, \,\, \textrm{(\ref{PF1CST2})}, \,\,  \textrm{(\ref{PF1CST3})}, \,\, \textrm{(\ref{PF2CST1})}, \nonumber
\end{align} 
\end{subequations}
where $\Vert \mathbf{P} \Vert_{\sf{F}}^{2}$ is the square of the Frobenius norm of $\mathbf{P}$, which denotes the total transmission power. 
{ This term penalizes excessive transmit power usage at the LEO satellite, thereby reducing overall power consumption while maintaining effective RM performance.
To be more specific, by adding $\Vert \mathbf{P} \Vert_{\sf{F}}^{2}$ to the objective function while preserving the convexity, the LEO satellite no longer transmits at the maximum allowable power levels. Instead, it would increase the transmission power level as long as the inequality (\ref{PF2CST1}) holds with equality (i.e., $R_{{\sf{p}},k} = \alpha_{{\sf{p}},k}$), preventing the wasted offered rates.}
In addition, it is worth pointing out that the power efficiency can be enhanced by adjusting the regularization parameter $\eta$ to a reasonable value for the desired RM.
That is, $\eta$ can be flexibly adjusted according to the traffic demand of users so that the satellite can efficiently use the transmission power.
By doing so, the sustainability of satellite networks, which are highly power-limited, can be increased.
However, $\mathscr{P}_3$ is still a non-convex problem, owing to the constraint (\ref{PF1CST1}) and first equation of (\ref{PF2CST1}). To convexify these constraints, slack variables $a_{{\sf{c}},k}$ and $b_{{\sf{p}},k}$ are further introduced as follows:
\begin{align}
\nonumber
\mathscr{P}_4: \,\,\,\, 
\minimize_{\mathbf{P}, \mathbf{c}, \boldsymbol{\alpha}_{\sf{p}}, \mathbf{a}_{\sf{c}}, \mathbf{b}_{\sf{p}}} \,
\eta{\Vert {\mathbf{r}_{\sf target}- (\mathbf{c} + \boldsymbol{\alpha}_{\sf{p}})} \Vert}_{2}^{2} + (1-\eta) \Vert {\mathbf{P}} \Vert_{\sf{F}}^{2}
\end{align}\setcounter{equation}{22}
\begin{subequations}\label{condition}
\begin{align}
\mathrm{s.t.}\,\,\,\,\,\,
\label{PF4CST1}
&\log_{2}{(1+a_{{\sf{c}},k})} \geq \sum_{j=1}^{K}{C_j}, \,\, \forall k \in \mathcal{K}, \\
\label{PF4CST2}
&\frac{\mathbf{p}_{\sf{c}}^{\sf{H}}(\hat{\mathbf{h}}_{k} \hat{\mathbf{h}}_{k}^{\sf{H}} \odot
    [e^{-\delta_{\sf{fb}}^{2}}\mathbf{1}_{N_{{\sf{t}}}} + (1-e^{-\delta_{\sf{fb}}^{2}})\mathbf{I}_{N_{{\sf{t}}}}])\mathbf{p}_{\sf{c}}}
    {l_{{\sf{c}},k} + \sigma_{n}^{2}} \geq a_{{\sf{c}},k},
\nonumber \\
& \forall k \in \mathcal{K},\\
\label{PF4CST3}
&\log_{2}{(1+b_{{\sf{p}},k})} \geq \alpha_{{\sf{p}},k},   \,\, \forall k \in \mathcal{K},\\
\label{PF4CST4}
&\frac{\mathbf{p}_{k}^{\sf{H}}(\hat{\mathbf{h}}_{k} \hat{\mathbf{h}}_{k}^{\sf{H}} \odot
    [e^{-\delta_{\sf{fb}}^{2}}\mathbf{1}_{N_{{\sf{t}}}} + (1-e^{-\delta_{\sf{fb}}^{2}})\mathbf{I}_{N_{{\sf{t}}}}])\mathbf{p}_{k}}
    {l_{{\sf{p}},k} + \sigma_{n}^{2}} \geq b_{{\sf{p}},k},
\nonumber \\
& \forall k \in \mathcal{K},\\
\label{PF4CST5}
&a_{{\sf{c}},k} \geq 0, b_{{\sf{p}},k} \geq 0, \alpha_{{\sf{p}},k} \geq 0, \,\, \forall k \in \mathcal{K},\\
& \textrm{(\ref{PF1CST2})}, \,\,  \textrm{(\ref{PF1CST3})}, \nonumber 
\end{align}
\end{subequations}
where $\mathbf{a}_{\sf{c}}=[a_{{\sf{c}},1},\cdots,a_{{\sf{c}},K}]^{\sf{T}}$ and $\mathbf{b}_{\sf{p}}=[b_{{\sf{p}},1},\cdots,b_{{\sf{p}},K}]^{\sf{T}}$. In the reformulated problem $\mathscr{P}_{4}$, owing to the constraints (\ref{PF4CST1}) and (\ref{PF4CST2}), $R_{{\sf{c}}, k} \geq \log_{2}{(1+a_{{\sf{c}}, k})} \geq \sum_{j=1}^{K}{C_{j}}$ holds. Moreover, owing to constraints (\ref{PF4CST3}) and (\ref{PF4CST4}), $R_{{\sf{p}}, k} \geq \log_{2}{(1+b_{{\sf{p}}, k})} \geq \alpha_{{\sf{p}},k}$ holds. Therefore, $\mathscr{P}_4$ preserves $\mathscr{P}_3$. 
However, the reformulated problem $\mathscr{P}_4$ is still a non-convex problem {due to} the constraints (\ref{PF4CST2}) and (\ref{PF4CST4}). To convexify these constraints, we first expand the constraint (\ref{PF4CST2}) as 
\begin{align}\label{C_SINR_Eq}
f_{1}{(\mathbf{p_{\sf{c}}}, a_{{\sf{c}},k})} & \triangleq \frac{\mathbf{p}_{\sf{c}}^{\sf{H}}(\hat{\mathbf{h}}_{k} \hat{\mathbf{h}}_{k}^{\sf{H}} \odot
    [e^{-\delta_{\sf{fb}}^{2}}\mathbf{1}_{N_{{\sf{t}}}} + (1-e^{-\delta_{\sf{fb}}^{2}})\mathbf{I}_{N_{{\sf{t}}}}])\mathbf{p}_{\sf{c}}}
    {a_{{\sf{c}},k}} 
    \nonumber \\
    & \geq l_{{\sf{c}},k} + \sigma_{n}^{2}.
\end{align}

\begin{figure*}[!t]
\begin{align}\label{C_SCA_Eq}
\small
\hat{f}_{1}{(\mathbf{p}_{\sf{c}}, a_{{\sf{c}},k} ; \mathbf{p}_{\sf{c}}^{[n]}, a_{{\sf{c}},k}^{[n]})}  \triangleq 
&2\text{Re}\bigg[\frac{(\mathbf{p}_{\sf{c}}^{[n]})^{\sf{H}}(\hat{\mathbf{h}}_{k} \hat{\mathbf{h}}_{k}^{\sf{H}} \odot[e^{-\delta_{\sf{fb}}^{2}}\mathbf{1}_{N_{{\sf{t}}}}\! +\! (1\!-\!e^{-\delta_{\sf{fb}}^{2}})\mathbf{I}_{N_{{\sf{t}}}}])\mathbf{p}_{\sf{c}}}{a_{{\sf{c}},k}^{[n]}}\bigg]  \nonumber \\
& - \frac{(\mathbf{p}_{\sf{c}}^{[n]})^{\sf{H}}(\hat{\mathbf{h}}_{k} \hat{\mathbf{h}}_{k}^{\sf{H}} \odot[e^{-\delta_{\sf{fb}}^{2}}\mathbf{1}_{N_{{\sf{t}}}}\! +\! (1\!-\!e^{-\delta_{\sf{fb}}^{2}})\mathbf{I}_{N_{{\sf{t}}}}])\mathbf{p}_{\sf{c}}^{[n]}}{(a_{{\sf{c}},k}^{[n]})^2} a_{{\sf{c}},k}.
\end{align}
\vspace{-7mm}
\end{figure*}
\begin{figure*}[!t]
\begin{align} \label{P_SCA_Eq}
\small
\hat{f}_{2}{(\mathbf{p}_{k}, b_{{\sf{p}},k} ; \mathbf{p}_{k}^{[n]}, b_{{\sf{p}},k}^{[n]})}    \triangleq
2\text{Re}\bigg[\frac{(\mathbf{p}_{k}^{[n]})^{\sf{H}}(\hat{\mathbf{h}}_{k} \hat{\mathbf{h}}_{k}^{\sf{H}} \odot[e^{-\delta_{\sf{fb}}^{2}}\mathbf{1}_{N_{{\sf{t}}}}\! +\! (1\!-\!e^{-\delta_{\sf{fb}}^{2}})\mathbf{I}_{N_{{\sf{t}}}}])\mathbf{p}_{k}}{b_{{\sf{p}},k}^{[n]}}\bigg] \nonumber \\
- \frac{(\mathbf{p}_{k}^{[n]})^{\sf{H}}(\hat{\mathbf{h}}_{k} \hat{\mathbf{h}}_{k}^{\sf{H}} \odot[e^{-\delta_{\sf{fb}}^{2}}\mathbf{1}_{N_{{\sf{t}}}}\! +\! (1\!-\!e^{-\delta_{\sf{fb}}^{2}})\mathbf{I}_{N_{{\sf{t}}}}])\mathbf{p}_{k}^{[n]}}{(b_{{\sf{p}},k}^{[n]})^2} b_{{\sf{p}},k}.
\end{align}
\noindent\rule{\textwidth}{.5pt}
\end{figure*}
Note that because $l_{{\sf{c}}, k} + \sigma_{n}^2 \geq 0$ and $a_{{\sf{c}}, k} \geq 0$, expanding the equation as done above does not cause a loss of generality.
Herein, the function $f_{1}{(\mathbf{p}_{\sf{c}}, a_{{\sf{c}},k})}$ defined at the {left-hand side (LHS)} of equation (\ref{C_SINR_Eq}) is a quadratic over linear function, which is a convex function with respect to $[\mathbf{p}_{\sf{c}}^{\sf{T}}, a_{{\sf{c}},k}]^{\sf{T}}$. Hence, following the first-order condition property of a convex function, $f_{1}{(\mathbf{p_{\sf{c}}}, a_{{\sf{c}},k})}$ can be approximated and lower bounded as an affine function using the first-order Taylor expansion at point $[(\mathbf{p_{\sf{c}}}^{[n]} )^{\sf{T}}, {a_{{\sf{c}},k}^{[n]}}]^{\sf{T}}$ as {(\ref{C_SCA_Eq}) at the top of this page.} 
With a similar manner, the constraint (\ref{PF4CST4}) can be expanded as
\begin{align}
\label{P_SINR_Eq}
f_{2}{(\mathbf{p}_{k}, b_{{\sf{p}},k})} & \triangleq \frac{\mathbf{p}_{k}^{\sf{H}}(\hat{\mathbf{h}}_{k} \hat{\mathbf{h}}_{k}^{\sf{H}} \odot [e^{-\delta_{\sf{fb}}^{2}}\mathbf{1}_{N_{{\sf{t}}}} + (1-e^{-\delta_{\sf{fb}}^{2}})\mathbf{I}_{N_{{\sf{t}}}}])\mathbf{p}_{k}} {b_{{\sf{p}},k}}
\nonumber \\
& \geq l_{{\sf{p}},k} + \sigma_{n}^{2}.
\end{align} 
Subsequently, $f_{2}{(\mathbf{p}_{k}, b_{{\sf{p}},k})}$ can be approximated and lower bounded at $[(\mathbf{p}_{k}^{[n]})^{\sf{T}},{b_{{\sf{p}},k}^{[n]}}]^{\sf{T}}$ as {(\ref{P_SCA_Eq}) at the top of this page.} 

In both (\ref{C_SCA_Eq}) and (\ref{P_SCA_Eq}), $n$ represents the $n$-th SCA iteration. Finally, constraints (\ref{PF4CST2}) and (\ref{PF4CST4}) can be near equivalently approximated as follows:
 \begin{align}
 \label{C_SCA_AP}
 &\hat{f}_{1}{(\mathbf{p}_{\sf{c}}, a_{{\sf{c}},k};\mathbf{p}_{\sf{c}}^{[n]}, a_{{\sf{c}},k}^{[n]})} \geq l_{{\sf{c}},k} + \sigma_{n}^{2}, \\
 \label{P_SCA_AP}
 &\hat{f}_{2}{(\mathbf{p}_{k}, b_{{\sf{p}},k} ; \mathbf{p}_{k}^{[n]}, b_{{\sf{p}},k}^{[n]})} \geq l_{{\sf{p}},k} + \sigma_{n}^{2}.
 \end{align}
These are convex constraints according to the following lemma. 

\begin{lemma}
$\hat{\mathbf{h}}_{k} \hat{\mathbf{h}}_{k}^{\sf{H}} \odot
    [e^{-\delta_{\sf{fb}}^{2}}\mathbf{1}_{N_{{\sf{t}}}} + (1-e^{-\delta_{\sf{fb}}^{2}})\mathbf{I}_{N_{{\sf{t}}}}]$ and $\hat{\mathbf{h}}_{k} \hat{\mathbf{h}}_{k}^{\sf{H}} \odot
    [(1-e^{-\frac{\delta_{\sf{ce}}^{2}}{2}})^{2}\mathbf{1}_{N_{{\sf{t}}}} + (1-e^{-\delta_{\sf{ce}}^{2}})\mathbf{I}_{N_{{\sf{t}}}}]$ are PSD matrices.
\end{lemma}
\begin{proof}
    {Please refer to Appendix A.} 
\end{proof}
Furthermore, to effectively handle constraints (\ref{PF4CST1}) and (\ref{PF4CST3}) with low complexity using the CVX toolbox, these constraints can be equivalently approximated into a second-order cone (SOC) constraint, as described in \cite{yin2022rate}. 
By doing so, the constraint (\ref{PF4CST1}) can be approximated as
\begin{align}
    \label{C_SOC_Eq}
    & a_{{\sf{c}},k} + v_{{\sf{c}},k}^{[n]} - \sum_{j=1}^{K}{C_{j}\ln{2}}
    \nonumber \\ 
    & \geq \bigg \Vert \bigg[ 2 \sqrt{u_{{\sf{c}},k}^{[n]}} \,\,\, , \,\,\, a_{{\sf{c}},k}-v_{{\sf{c}},k}^{[n]} + \sum_{j=1}^{K}{C_{j}\ln{2}} \bigg] \bigg \Vert_{2}, 
\end{align}
where $v_{{\sf{c}},k}^{[n]} = \ln{(1 + a_{{\sf{c}},k}^{[n]})} + \frac{a_{{\sf{c}},k}^{[n]}}{1+a_{{\sf{c}},k}^{[n]}}$ and 
$u_{{\sf{c}},k}^{[n]} = \frac{(a_{{\sf{c}},k}^{[n]})^2}{1+a_{{\sf{c}},k}^{[n]}}$. 
In a similar manner, the constraint (\ref{PF4CST3}) also can be near equivalently approximated into a SOC constraint as
\begin{align}
    \label{P_SOC_Eq}
    & b_{{\sf{p}},k} + v_{{\sf{p}},k}^{[n]} - \alpha_{{\sf{p}},k} \ln{2}
    \nonumber \\
    & \geq \bigg \Vert \bigg[ 2 \sqrt{u_{{\sf{p}},k}^{[n]}} \,\,\, , \,\,\, b_{{\sf{p}},k}-v_{{\sf{p}},k}^{[n]} + \alpha_{{\sf{p}},k}\ln{2} \bigg] \bigg \Vert_{2},
\end{align}
where $v_{{\sf{p}},k}^{[n]} = \ln{(1 + b_{{\sf{p}},k}^{[n]})} + \frac{b_{{\sf{p}},k}^{[n]}}{1+b_{{\sf{p}},k}^{[n]}}$
and $u_{{\sf{p}},k}^{[n]} = \frac{(b_{{\sf{p}},k}^{[n]})^2}{1+b_{{\sf{p}},k}^{[n]}}$.

In conclusion, a joint optimization problem for matching the uneven traffic demands and minimizing the transmission power based on the RSMA in the presence of imperfect CSIT and CSIR is reformulated into a convex problem as follows:
\begin{align}
\nonumber
\mathscr{P}_5: \,\,\,\, 
\minimize_{\mathbf{P}, \mathbf{c}, \boldsymbol{\alpha}_{{\sf{p}}}, \mathbf{a}_{\sf{c}}, \mathbf{b}_{{\sf{p}}}} \,
\eta{\Vert {\mathbf{r}_{\sf target}- (\mathbf{c} + \boldsymbol{\alpha}_{\sf{p}})} \Vert}_{2}^{2} + (1-\eta) \Vert {\mathbf{P}} \Vert_{\sf{F}}^{2}
\end{align}\setcounter{equation}{31}
\begin{subequations}\label{condition}
\begin{align}
\mathrm{s.t.}\,\,\,\,\,\,
& \textrm{(\ref{PF1CST2})}, \,\,  \textrm{(\ref{PF1CST3})}, \,\, \textrm{(\ref{PF4CST5})}, \,\, \textrm{(\ref{C_SCA_AP})}, \,\, \textrm{(\ref{P_SCA_AP})}, \,\, \textrm{(\ref{C_SOC_Eq})}, \,\, \textrm{(\ref{P_SOC_Eq})}. \nonumber
\end{align}
\end{subequations}
This reformulated problem $\mathscr{P}_5$ can be solved effectively using the CVX toolbox \cite{grant2014cvx}. The detailed procedure of our proposed scheme is summarized in \textbf{Algorithm \ref{Algorithm 1}}.

\begin{algorithm}[!t]
\caption{SCA-Based Power-Efficient RM Framework}\label{Algorithm 1}
\begin{algorithmic}
\State \textbf{Input}: $P_{t}$, $\mathbf{r}_{\sf{target}}$, $\eta$, $\hat{\mathbf{H}}$, $\delta_{\sf{fb}}$, $\delta_{\sf{ce}}$, $\epsilon$, $N$
\State \textbf{Initialize}: $\mathbf{P}^{[0]}$, $\mathbf{a}_{{\sf{c}}}^{[0]}$, $\mathbf{b}_{{\sf{p}}}^{[0]}$, $D^{[0]}$, $\epsilon$, $N$, $n \leftarrow 0$
\Repeat
      \State $n \leftarrow n+1$ 
      \State Solve $\mathscr{P}_5$ and obtain a solution $\mathbf{x}^{[\star]} = (\mathbf{P}^{[\star]}$, $\mathbf{c}^{[\star]}$, $\boldsymbol{\alpha}_{{\sf{p}}}^{[\star]}$, $\mathbf{a}_{{\sf{c}}}^{[\star]}$, $\mathbf{b}_{{\sf{p}}}^{[\star]})$
      \State Update  $\mathbf{P}^{[n]} \leftarrow \mathbf{P}^{[\star]}, \,\, \mathbf{a}_{{\sf{c}}}^{[n]} \leftarrow \mathbf{a}_{{\sf{c}}}^{[\star]}, \,\, \mathbf{b}_{{\sf{p}}}^{[n]} \leftarrow \mathbf{b}_{{\sf{p}}}^{[\star]}$
\Until{$\vert D^{[n]}-D^{[n-1]} \vert \leq \epsilon \,\, {\rm{or}} \,\, n \geq N$}
\State \textbf{Return}: $\mathbf{P}^{[\star]}$, $\mathbf{c}^{[\star]}$
\State \textbf{Output}: Calculate instantaneous total rate, $R_{k}^{[\star]}$ using $\mathbf{P}^{[\star]}$, $\mathbf{c}^{[\star]}$, $ \forall k \in \mathcal{K}$
\end{algorithmic}
\end{algorithm}
In \textbf{Algorithm \ref{Algorithm 1}}, $D^{[n]}$, $N$, and $\epsilon$ denote the value of ${\Vert \mathbf{r}_{\sf{target}}-(\mathbf{c} + \boldsymbol{\alpha}_{\sf{p}}) \Vert}_{2}^2$ at the $n$-th SCA iteration, maximum number of SCA iterations, and tolerance value, respectively. 
In other words, \textbf{Algorithm \ref{Algorithm 1}} terminates when the absolute value of the difference between $D^{[n]}$ and $D^{[n-1]}$ is lower than $\epsilon$, or the SCA iteration number $n$ becomes $N$. 
Note that as mentioned in \textbf{Algorithm \ref{Algorithm 1}}, we compute the proposed RSMA-based RM precoder, consisting of $\mathbf{P}^{[\star]}$ and $\mathbf{c}^{[\star]}$, by solving $\mathscr{P}_{5}$ iteratively, and then calculate the instantaneous rate $R_{k}^{[\star]}$, $ \forall k \in \mathcal{K}$, using the obtained precoder. 

\begin{remark} 
{\rm \textbf{(Key difference from power minimization problems with QoS constraints)}: 
{In power minimization problems with QoS constraints, infeasible solutions can result if the available transmit power budget is insufficient or the user channel condition is unfavorable. This, in turn, significantly undermines the communication reliability of multibeam LEO SATCOM, a system inherently constrained by power limitations, thus hindering its effectiveness in real-world applications.
The proposed RM framework, on the other hand, addresses these challenges by consistently ensuring feasible solutions through the adaptive matching of actual offered rates to traffic demands. Therefore, it guarantees reliable communication services regardless of the available power budget or user channel conditions. 
Furthermore, in ideal scenarios, where sufficient power and favorable channel conditions are available, our framework reduces to the conventional power minimization problem with QoS constraints, demonstrating its versatility and generality across various operating conditions.}}
\end{remark}

\begin{remark}
{\rm\textbf{(Complexity analysis)}: The constraints in $\mathscr{P}_{5}$ can be represented by SOC constraints so that $\mathscr{P}_{5}$ can be solved efficiently using a standard interior-point method. Thus, the computational complexity of the proposed RSMA-based RM algorithm in a big-O sense is given by $\mathcal{O}( N ( K N_{{\sf{t}}} )^{3.5})$ \cite{ye2011interior}.}
\end{remark}

\begin{remark}
{\rm\textbf{(Convergence analysis)}: In $\mathscr{P}_{5}$, an obtained solution at the $n$-th iteration $\mathbf{x}^{[n]}=(\mathbf{P}^{[n]}, \mathbf{c}^{[n]}, \boldsymbol{\alpha}_{{\sf{p}}}^{[n]}, \mathbf{a}_{{\sf{c}}}^{[n]}, \mathbf{b}_{{\sf{p}}}^{[n]})$ is included in the feasible set of the $(n+1)$-th iteration. Thus, the proposed algorithm generates a non-increasing sequence as the number of iterations increases. Moreover, since the objective function of $\mathscr{P}_{5}$ is the summation of the L2-norm and Frobenius norm, which is always greater or equal to zero, the objective value of $\mathscr{P}_{5}$ is bounded below. {Thanks to these features, the proposed SCA-based RM algorithm is guaranteed to converge to the set of stationary points of $\mathscr{P}_1$ \cite{marks1978general}.}}

\end{remark}

\section{Performance Evaluation}

In this section, the performance of the proposed RSMA-based power-efficient RM framework (denoted as ``{\sf{RM-RSMA}}'') is evaluated based on simulation results.
{ In particular, the proposed scheme is compared to the other schemes for satisfying the individual user traffic demands, which are the SDMA scheme (denoted as ``{\sf{RM-SDMA}}''), MMSE-based RSMA scheme (denoted as ``{\sf{MMSE-RSMA}}'')\cite{cui2023energy}, multicast-based RSMA scheme (denoted as ``{\sf{Multicast-RSMA}}'')\cite{10304489}, and four-color frequency reuse scheme (denoted as ``{\sf{RM-4color}}'').} 
Specifically, ``{\sf{MMSE-RSMA}}'' designs normalized private precoding vectors via the MMSE method while optimizing the power ratios of the private messages and the common precoding vector \cite{cui2023energy}. 
{``{\sf{Multicast-RSMA}}'' is a method based on the RSMA-based multibeam multicast transmission from \cite{10304489}, which is modified to meet the traffic demands of individual users.}
``{\sf{RM-4color}}'' denotes the conventional four-color reuse scheme in which the satellite partially reuses an available frequency band by dividing it into four sub-bands \cite{perez2019signal}. The proposed scheme is also compared with the conventional technique, namely, the RSMA-based MMF scheme (denoted as ``{\sf{MMF-RSMA}}'') that aims to improve the quality of the entire network \cite{yin2022rate}. 
Our simulation results consider both the perfect CSIT and CSIR case and the imperfect CSIT and CSIR case with phase perturbation due to channel estimation and feedback errors.  
\begin{figure}[!b]
\centering
 		\includegraphics[width=0.8\linewidth]{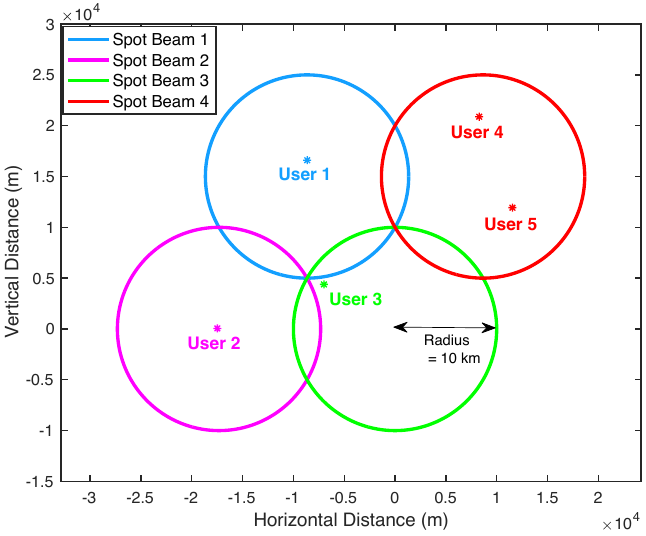}
 		\caption{{Beam pattern of four spot beams and location of users uniformly distributed within each beam.}}
    	\label{Fig_beam} 
\end{figure}
{
The channel model is set as introduced in Section \uppercase\expandafter{\romannumeral2}. 
The altitude of the satellite is assumed to be $\num{600}$ km, where LEO satellites are typically deployed, and the related simulation parameters are set according to 3GPP NTN standards \cite{3gpp_ntn}, as presented in Table  \ref{Table1}.
Given the effective isotropically radiated power (EIRP) density, calculated as $P_{\sf{t}}G_{\sf{max}}/N_{\sf{t}}B$, the per-feed transmit power budget is given by $P_{\sf{t}}/N_{\sf{t}} = 21.52$ dBm.
Moreover, the rain fading parameters $(\mu, \sigma) = (-2.6, 1.63)$ are set to reflect a temperate central European climate \cite{zheng2012generic}.
The noise variance, regularization parameter, maximum iteration number, and tolerance value are set to be $\sigma_n^2 = 1$, $\eta = 0.91$ bps/Hz/W, $N = 20$, and $\epsilon = 10^{-4}$, respectively. 
The number of feeds and users are set to $N_{\sf{t}} = 4$ and $K = 5$, respectively, representing an overloaded scenario that frequently arises in multibeam LEO SATCOM.} 
The beam pattern of four spot beams and the locations of the users are illustrated in Fig. \ref{Fig_beam}. The simulation results are obtained by averaging 100 channel realizations.

\begin{table}[!t]\renewcommand{\arraystretch}{1} 
\centering
{
\caption{{Simulation Parameters \cite{3gpp_ntn}}}
\label{Table1}
\centering
\begin{tabular}{|c|c|c|}
\hlineB{3}
\textbf{Abbreviation} & \textbf{Definition} & \textbf{Value} \\
\hlineB{3}
\hhline{|---|}
{$f_{\sf{c}}$} & {Frequency band (Carrier frequency)} & {Ka ($20$ GHz)}\\
{$B$} & {Bandwidth} & {$400$ MHz}\\
{$\theta_{\mathrm{3dB}}$} & {3 dB angle} & {$1.7647^{\circ}$}\\
{$G_{\sf{max}}$} & {Maximum beam gain} & {$38.5$ dBi}\\
{$G_{\sf{R}}$} & {User terminal antenna gain} & {$39.7$ dBi}\\
{$T_{\sf{sys}}$} & {Noise temperature} & {$150$ K}\\
{$-$} & {EIRP density} & {$4$ dBW/MHz}\\
{$-$} & {Satellite altitude} & {$600$ km}\\
{$-$} & {Beam diameter} & {$20$ km}\\
\hlineB{3}
\end{tabular}}
\end{table}

The achievable rates of each scheme under the perfect CSIT and CSIR and the imperfect CSIT and CSIR are compared in Fig. \ref{result_1} and  Fig. \ref{result_3}, respectively. {The traffic demand $\mathbf{r}_{\sf{target}}$ is assumed to be $[2, 2, 3, 3.5, 4]^{\sf{T}}$ bps/Hz in both cases.} 
As shown in the figures, the white bar at the far left of each user indicates the traffic demand of each user.
The proposed ``{\sf{RM-RSMA}}'' scheme satisfies the traffic demands of each user much better than other benchmark schemes in both cases. The superiority of ``{\sf{RM-RSMA}}'' 
turns up owing to the occurring tendency while satisfying the uneven traffic demands. This is because all the schemes except ``{\sf{MMF-RSMA}}'' first tend to fulfill the traffic demands of the users with comparatively high demands, as their objective is to minimize the summation of the users' {unused} or {unmet rate}.
Thus, as the {spatial degrees of freedom} is limited in the user-overloaded case, ``{\sf{RM-SDMA}}'' causes significant interference to other users while satisfying the requirements of users with higher traffic demands.

\begin{figure}[!t]
\centering
 		\includegraphics[width=0.8\linewidth]{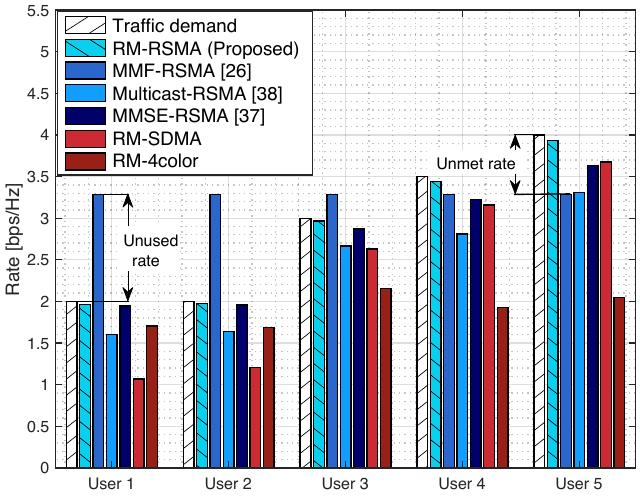}
 		\caption{{Achievable rate comparison of each user under perfect CSIT and CSIR ($\delta_{\sf{fb}}=0^{\circ}, \delta_{\sf{ce}}=0^{\circ}$) when the traffic demand $\mathbf{r}_{\sf{target}}$ is $[2, 2, 3, 3.5, 4]^{\sf{T}}$.}}
    	\label{result_1} 
\end{figure}

\begin{figure}[!t]
\centering
 		\includegraphics[width=0.8\linewidth]{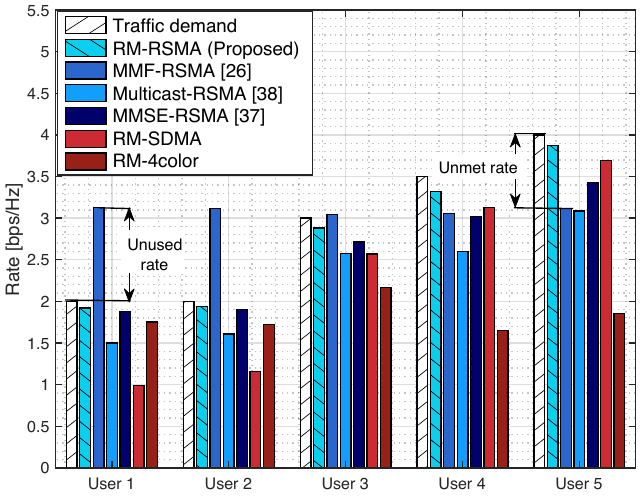}
 		\caption{{Achievable rate comparison of each user under imperfect CSIT and CSIR ($\delta_{\sf{fb}}=5^{\circ}, \delta_{\sf{ce}}=2^{\circ}$) when the traffic demand $\mathbf{r}_{\sf{target}}$ is $[2, 2, 3, 3.5, 4]^{\sf{T}}$.}}
    	\label{result_3} 
\end{figure}

{To be more specific, when the traffic demand is $\mathbf{r}_{\sf{target}}=[2, 2, 3, 3.5, 4]^{\sf{T}}$, the users in spot beams 3 and 4 have comparatively high traffic demands than the other users.} In this case, ``{\sf{RM-SDMA}}'' first tends to decrease the effect of inter-beam interference as well as intra-beam interference on spot beam 4 to satisfy the larger traffic demands over a certain level. 
For the same reason, it tends to decrease the effect of inter-beam interference on spot beam 3. 
However, by doing so, spot beams 1 and 2 suffer from severe inter-beam interference from the other beams. 
That is to say, the requirements of the users in spot beams 1 and 2, whose traffic demands are comparatively low, cannot be properly fulfilled. 
Further, the traffic demands of users in spot beams 3 and 4 are also not fully fulfilled yet.
On the contrary, ``{\sf{RM-RSMA}}'' partially fulfills the comparatively high traffic demands by using the common stream, as demonstrated in Fig. \ref{result_2} and  Fig. \ref{result_4}.
{Therefore, compared to ``{\sf{RM-SDMA}}'', the proposed ``{\sf{RM-RSMA}}'' allows users with lower traffic demands to experience relatively less interference, while effectively fulfilling the remaining unmet rates of users in spot beams 3 and 4 that are not yet satisfied.
By doing so, the entire traffic demands are satisfied uniformly.}
Notably, even though `{\sf{MMSE-RSMA}}'' is also based on the RSMA, a performance gap regarding {unused} or {unmet rate} occurs. This is because ``{\sf{RM-RSMA}}'' can more flexibly design private precoding vectors according to the traffic demands, network load status, and channel status than ``{\sf{MMSE-RSMA}}'', the normalized private precoding vectors of which are fixed to the MMSE. 
{Moreover, ``{\sf{Multicast-RSMA}}'' struggles to meet the traffic demands of individual users when the number of users within a beam is large and their demands are high, resulting in suboptimal performance compared to the proposed method.}
``{\sf{RM-4color}}'' exhibits the worst performance among the compared schemes as it cannot reuse the whole frequency band. 
Meanwhile, ``{\sf{MMF-RSMA}}'' incurs large {unused} and {unmet rates} due to its oversight of uneven traffic distributions.

\begin{figure}[!t]
\centering
 		\includegraphics[width=0.8\linewidth]{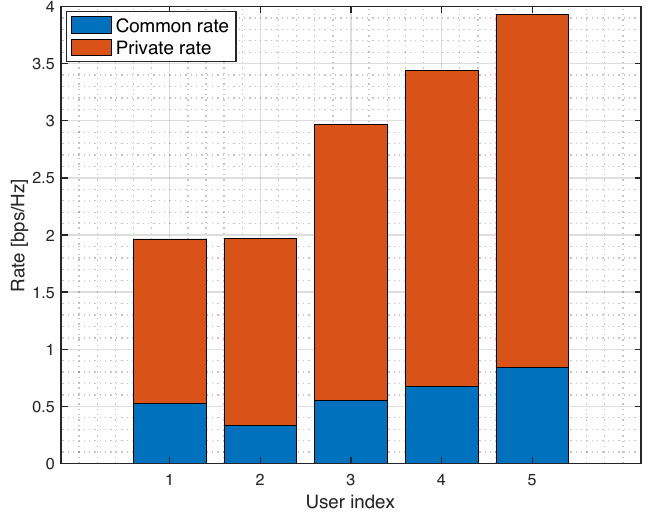}
 		\caption{{Portion of the common rate and private rate of each user for ``{\sf{RM-RSMA}}'' under perfect CSIT and CSIR ($\delta_{\sf{fb}}=0^{\circ}, \delta_{\sf{ce}}=0^{\circ}$) when the traffic demand $\mathbf{r}_{\sf{target}}$ is $[2, 2, 3, 3.5, 4]^{\sf{T}}$.}}
    	\label{result_2} 
\end{figure}

\begin{figure}[!t]
\centering
 		\includegraphics[width=0.8\linewidth]{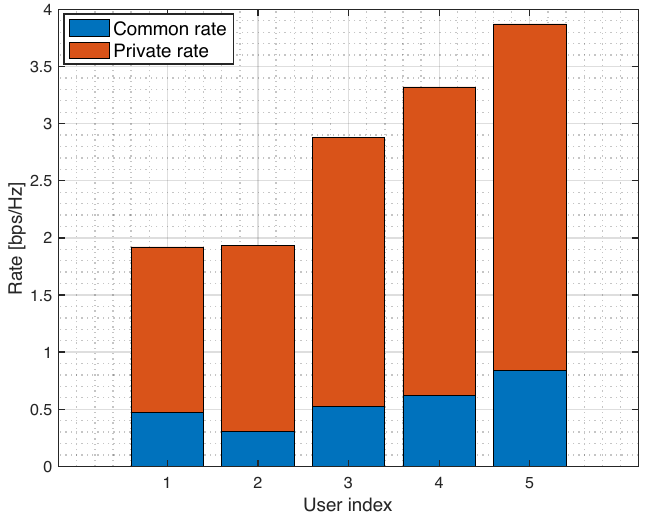}
 		\caption{{Portion of the common rate and private rate of each user for ``{\sf{RM-RSMA}}''  under imperfect CSIT and CSIR ($\delta_{\sf{fb}}=5^{\circ}, \delta_{\sf{ce}}=2^{\circ}$) when the traffic demand $\mathbf{r}_{\sf{target}}$ is $[2, 2, 3, 3.5, 4]^{\sf{T}}$.}}
    	\label{result_4} 
\end{figure}

{Thereafter, we present the traffic demand satisfaction in percentages through Fig. \ref{result_5} when the traffic demand is $\mathbf{r}_{\sf{target}} = [2, 2, 3, 3.5, 4]^{\sf{T}}$.}
The y-axis represents the ratio of robustness against the {unmet} and {unused rates}, a higher value indicating better performance. The blue (left in each scheme) and orange (right in each scheme) bars denote the average values for traffic demand satisfaction, and the black lines on each bar denote the distribution of traffic demand satisfaction. Herein, the blue bars denote the case of perfect CSIT and CSIR, namely, $\delta_{\sf{fb}} = 0^{\circ}$, $\delta_{\sf{ce}} = 0^{\circ}$, and the orange bars denote the case of imperfect CSIT and CSIR, namely, $\delta_{\sf{fb}} = 5^{\circ}$, $\delta_{\sf{ce}} = 2^{\circ}$. 
Overall, as shown in the figure, the proposed scheme shows the highest average traffic demand satisfaction as well as the lowest variance in traffic demand satisfaction compared to the other schemes. 
In other words, our proposed framework can stably satisfy the traffic demands of users by effectively reducing the {unmet} or {unused rate} compared to other schemes.

\begin{figure}[!t]
\centering
 		\includegraphics[width=0.8\linewidth]{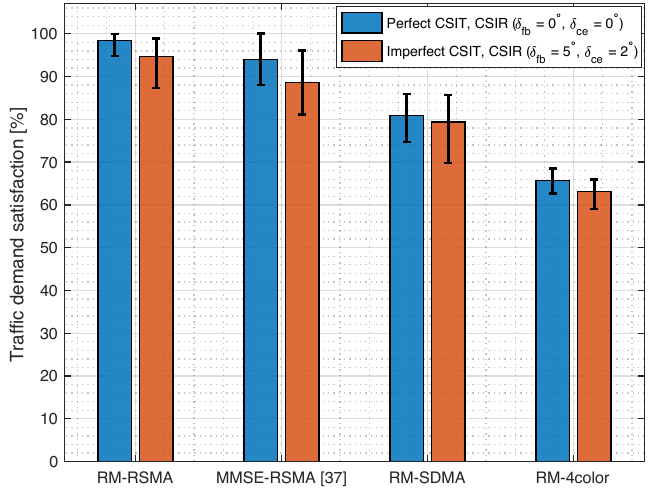}
 		\caption{{Traffic demand satisfaction comparison under both perfect CSIT and CSIR and imperfect CSIT and CSIR when the traffic demand $\mathbf{r}_{\sf{target}}$ is $[2, 2, 3, 3.5, 4]^{\sf{T}}$.}}
    	\label{result_5} 
\end{figure}

\begin{figure}[!t]
\centering
 		\includegraphics[width=0.8\linewidth]{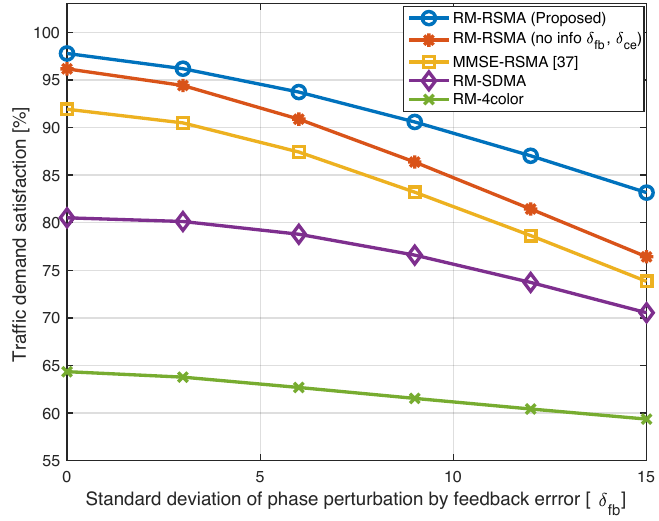}
 		\caption{{Traffic demand satisfaction comparison per $\delta_{\sf{fb}}$ when the traffic demand $\mathbf{r}_{\sf{target}}$ is $[2, 2, 3, 3.5, 4]^{\sf{T}}$. $\delta_{\sf{ce}}$ is set to  $\delta_{\sf{ce}}=2^{\circ}$.}}
    	\label{result_6} 
\end{figure}

This superiority of ``{\sf{RM-RSMA}}'' comes from the effective inter-/intra-beam interference mitigation of RSMA and suitable precoder design according to the heterogeneous traffic demands.
{To elaborate numerically, the average percentage value for the traffic demand satisfaction of ``{\sf{RM-RSMA}}'' is increased by 4.8 and 6.8 $\%$, compared to that of ``{\sf{MMSE-RSMA}}'' under perfect CSIT and CSIR and imperfect CSIT and CSIR, respectively.}
This is because as discussed previously, the proposed scheme appropriately designs both private and common precoding vectors according to the status of networks, such as traffic demands and user-overloaded. 
In contrast, ``{\sf{MMSE-RSMA}}'' designs normalized precoding vectors through the MMSE method without considering the status of networks. 
{The average value for the traffic demand satisfaction of ``{\sf{RM-RSMA}}'' is increased by 21.6 and 19.3 $\%$, compared to that of ``{\sf{RM-SDMA}}'' under perfect CSIT and CSIR and imperfect CSIT and CSIR, respectively.}
{Moreover, the average value for the traffic demand satisfaction of ``{\sf{RM-RSMA}}'' is increased by 49.9 and 50.1 $\%$, compared to that of ``{\sf{RM-4color}}'' under perfect CSIT and CSIR and imperfect CSIT and CSIR, respectively.}
For SDMA schemes, performance degradation occurs because the inter-/intra-beam interference is not effectively mitigated when compared to that in the RSMA schemes, as the {spatial degrees of freedom} is limited. 
Particularly, ``{\sf{RM-4color}}'' shows the lowest satisfaction percentage compared to other schemes since the frequency band can not be used effectively.

{ We now investigate traffic demand satisfaction as a function of CSIT accuracy, as shown in Fig. \ref{result_6}. The x-axis represents the standard deviation of phase perturbation due to channel feedback errors ($\delta_{\sf{fb}}$). Herein, we also compare the proposed ``{\sf{RM-RSMA}}'' scheme with the ``{\sf{RM-RSMA (no info $\delta_{\sf{fb}}$, $\delta_{\sf{ce}}$)}}'' scheme, which does not leverage statistical information on phase perturbations when designing the RSMA-based RM precoder.
From Fig. \ref{result_6}, it is evident that the proposed ``{\sf{RM-RSMA}}'' consistently outperforms all benchmarks, regardless of CSIT accuracy. Notably, the performance gap between ``{\sf{RM-RSMA}}'' and ``{\sf{RM-RSMA (no info $\delta_{\sf{fb}}$, $\delta_{\sf{ce}}$)}}'' becomes increasingly significant as $\delta_{\sf{fb}}$ increases. This underscores the crucial role of incorporating statistical information on phase perturbations to design a robust rate-matching precoder, allowing the proposed ``{\sf{RM-RSMA}}'' scheme to maintain superior performance even under higher levels of phase perturbation. Also, even without utilizing statistical information on phase perturbations, the ``{\sf{RM-RSMA (no info $\delta_{\sf{fb}}$, $\delta_{\sf{ce}}$)}}'' scheme still outperforms the ``{\sf{MMSE-RSMA}}'', ``{\sf{RM-SDMA}}'', and ``{\sf{RM-4color}}'' schemes across the entire range of $\delta_{\sf{fb}}$ values, further highlighting the importance of adaptability to uneven traffic demands.}

\begin{figure}[!t]
\centering
 		\includegraphics[width=0.8\linewidth]{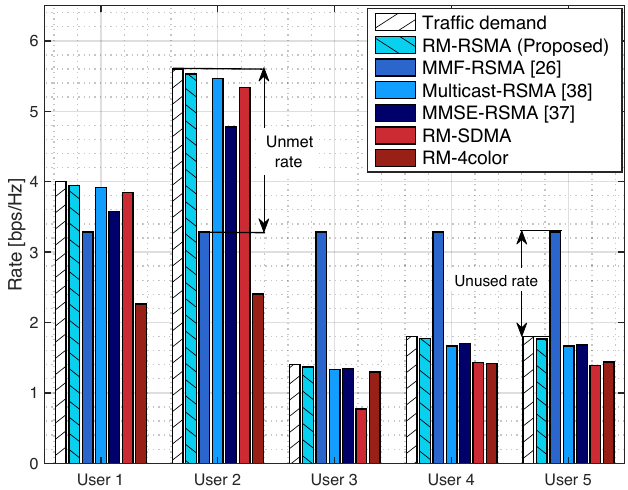}
 		\caption{{Achievable rate comparison of each user under perfect CSIT and CSIR ($\delta_{\sf{fb}}=0^{\circ}, \delta_{\sf{ce}}=0^{\circ}$) when the traffic demand $\mathbf{r}_{\sf{target}}$ is $[4, 5.6, 1.4, 1.8, 1.8]^{\sf{T}}$.}}
    	\label{result_7} 
\end{figure}

\begin{figure}[!t]
\centering
 		\includegraphics[width=0.8\linewidth]{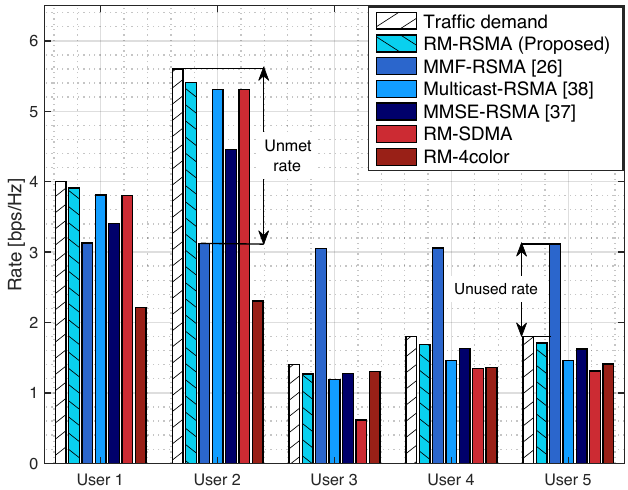}
 		\caption{{Achievable rate comparison of each user under imperfect CSIT and CSIR ($\delta_{\sf{fb}}=5^{\circ}, \delta_{\sf{ce}}=2^{\circ}$) when the traffic demand $\mathbf{r}_{\sf{target}}$ is $[4, 5.6, 1.4, 1.8, 1.8]^{\sf{T}}$.}}
    	\label{result_9} 
\end{figure}

\begin{figure}[!t]
\centering
 		\includegraphics[width=0.8\linewidth]{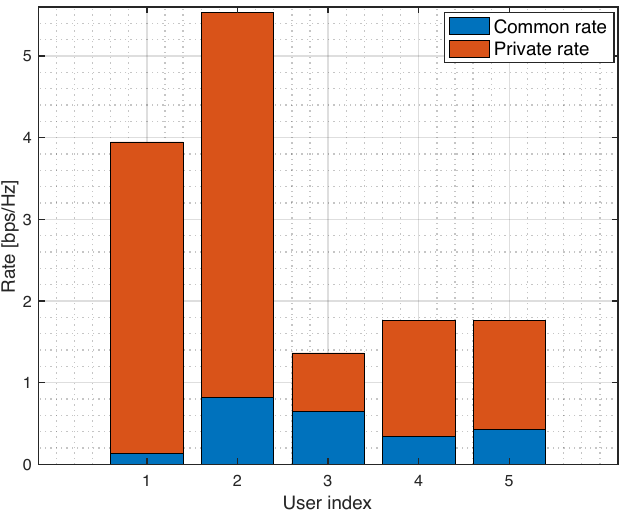}
 		\caption{{Portion of the common rate and private rate of each user for ``{\sf{RM-RSMA}}'' under perfect CSIT and CSIR ($\delta_{\sf{fb}}=0^{\circ}, \delta_{\sf{ce}}=0^{\circ}$) when the traffic demand $\mathbf{r}_{\sf{target}}$ is $[4, 5.6, 1.4, 1.8, 1.8]^{\sf{T}}$.}}
    	\label{result_8} 
\end{figure}

\begin{figure}[!t]
\centering
 		\includegraphics[width=0.8\linewidth]{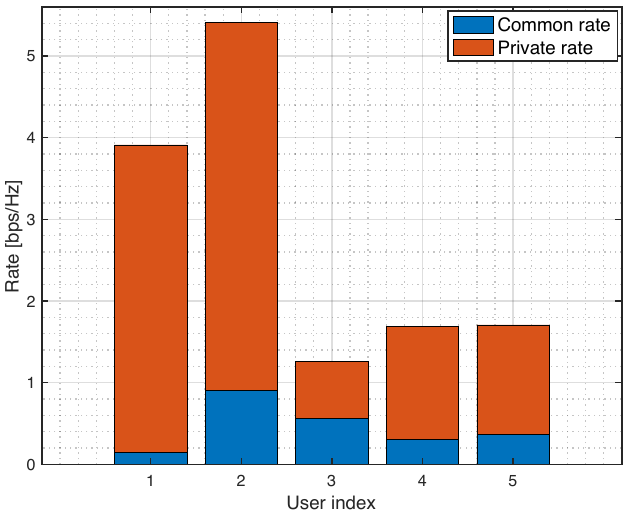}
 		\caption{{Portion of the common rate and private rate of each user for ``{\sf{RM-RSMA}}''  under imperfect CSIT and CSIR ($\delta_{\sf{fb}}=5^{\circ}, \delta_{\sf{ce}}=2^{\circ}$) when the traffic demand $\mathbf{r}_{\sf{target}}$ is $[4, 5.6, 1.4, 1.8, 1.8]^{\sf{T}}$.}}
    	\label{result_10} 
\end{figure}

{Furthermore, we provide the simulation results when the traffic demand is $\mathbf{r}_{\sf{target}}=[4, 5.6, 1.4, 1.8, 1.8]^{\sf{T}}$ bps/Hz, in which the users in spot beams 1 and 2 require comparatively larger traffic demands than others. }
First, the achievable rates of each scheme under the perfect CSIT and CSIR and the imperfect CSIT and CSIR are compared in Fig. \ref{result_7} and  Fig. \ref{result_9}, respectively. 
The results demonstrate that the proposed ``{\sf{RM-RSMA}}'' scheme exhibits better performance in the reduction of {unmet} and {unused rates}, compared to the other benchmark schemes for both cases.
To elaborate further, in this case, ``{\sf{RM-SDMA}}'' tends to decrease the effect of the inter-beam interference from other beams to spot beams 1 and 2 to satisfy the traffic demands of users 1 and 2 over a certain level at first.
However, by doing so, the effects of intra-beam and inter-beam interferences on spot beam 4 increase, which in turn, results in the traffic demands of the users in spot beam 4 being unsatisfied. 
For the same reason, the effect of inter-beam interference on spot beam 3 increases, because of which the traffic demand of user 3 cannot be properly satisfied. 
Furthermore, it can be observed that the traffic demands of users 1 and 2 are not completely fulfilled yet. 
These results are also caused by the limited {spatial degrees of freedom} of SDMA in the user-overloaded scenario.
On the other hand, ``{\sf{RM-RSMA}}'' partially fulfills the traffic demands of users by using the common stream as demonstrated in Fig. \ref{result_8} and Fig. \ref{result_10}. 
Therefore, through the common stream, the intra-beam interference within spot beam 4 is effectively mitigated, and the effect of inter-beam interference from the other beams to spot beams 3 and 4 is also significantly decreased. Additionally, the remaining {unmet rates} of users 1 and 2, which are not satisfied yet in ``{\sf{RM-SDMA}}'' is also fulfilled through the common stream.
``{\sf{MMSE-RSMA}}'' exhibits a larger {unused} or {unmet rate} than ``{\sf{RM-RSMA}}'' since it cannot flexibly design private precoding vectors according to the traffic demands. 
{Particularly, the performance gap between ``{\sf{RM-RSMA}}'' and ``{\sf{MMSE-RSMA}}'' increases, compared to the case when the traffic demand is
$\mathbf{r}_{\sf{target}} = [2,2,3,3.5,4]^{\sf{T}}$.} 
This is because ``{\sf{MMSE-RSMA}}'' cannot fulfill a high value of the traffic demands as shown in Fig. \ref{result_7} and  Fig. \ref{result_9}.
{For ``{\sf{Multicast-RSMA}}'', since the user traffic demands in beams with a large number of users are comparatively lower than those in the case of $\mathbf{r}_{\sf{target}} = [2,2,3,3.5,4]^{\sf{T}}$, it exhibits better performance in reducing unmet rates compared to the $\mathbf{r}_{\sf{target}} = [2,2,3,3.5,4]^{\sf{T}}$ case.}
Moreover, as ``{\sf{RM-4color}}'' cannot effectively use the frequency band, it exhibits the worst performance among the compared schemes.
In `{\sf{MMF-RSMA}}'', due to the overlook of heterogeneous traffic distributions, it results in large {unused} and {unmet rates}.

\begin{figure}[!t]
\centering
 		\includegraphics[width=0.8\linewidth]{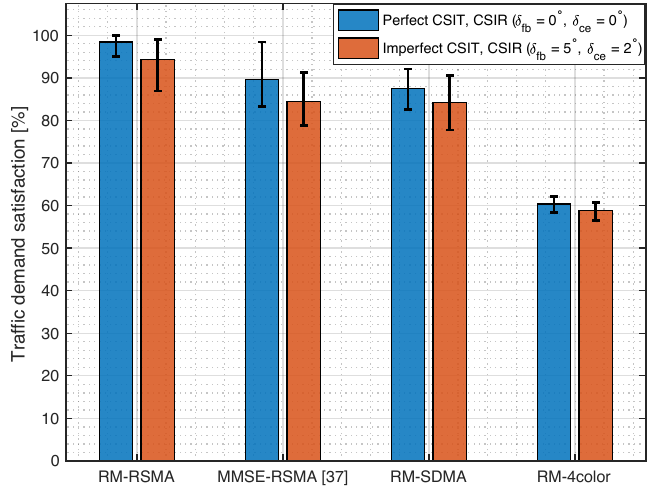}
 		\caption{{Traffic demand satisfaction comparison under both perfect CSIT and CSIR and imperfect CSIT and CSIR when the traffic demand $\mathbf{r}_{\sf{target}}$ is $[4, 5.6, 1.4, 1.8, 1.8]^{\sf{T}}$.}}
    	\label{result_11} 
\end{figure}

\begin{figure}[!t]
\centering
 		\includegraphics[width=0.8\linewidth]{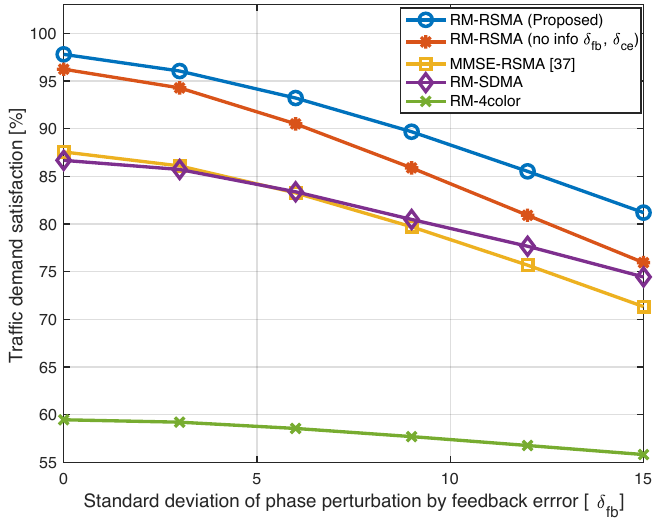}
 		\caption{{Traffic demand satisfaction comparison per $\delta_{\sf{fb}}$ when the traffic demand $\mathbf{r}_{\sf{target}}$ is $[4, 5.6, 1.4, 1.8, 1.8]^{\sf{T}}$. $\delta_{\sf{ce}}$ is set to  $\delta_{\sf{ce}}=2^{\circ}$.}}
    	\label{result_12} 
\end{figure}

{The traffic demand satisfaction is shown in Fig. \ref{result_11} when the traffic demand is $\mathbf{r}_{\sf{target}} = [4, 5.6, 1.4, 1.8, 1.8]^{\sf{T}}$.} Similar to Fig. \ref{result_5}, the blue (perfect CSIT and CSIR, i.e., $\delta_{\sf{fb}} = 0^{\circ}$, $\delta_{\sf{ce}} = 0^{\circ}$) and orange bars (imperfect CSIT and CSIR, i.e., $\delta_{\sf{fb}} = 5^{\circ}$, $\delta_{\sf{ce}} = 2^{\circ}$) denote the average values for traffic demand satisfaction, and the black lines on each bar denote the distribution of traffic demand satisfaction. 
Similar to previous results, ``{\sf{RM-RSMA}}'' stably satisfies the traffic demands of users compared to other schemes owing to its robust inter-/intra-beam interference management and flexible precoder design according to the traffic demands.
{Specifically, the average percentage value for the traffic demand satisfaction of ``{\sf{RM-RSMA}}'' is increased by 9.9 and 11.7 $\%$, compared to that of ``{\sf{MMSE-RSMA}}'' under perfect CSIT and CSIR and imperfect CSIT and CSIR, respectively.}
{Furthermore, the average value for traffic demand satisfaction of ``{\sf{RM-RSMA}}'' is increased by 12.6 and 11.8 $\%$, compared to that of ``{\sf{RM-SDMA}}'' under perfect CSIT and CSIR and imperfect CSIT and CSIR, respectively.} {The average value for traffic demand satisfaction of ``{\sf{RM-RSMA}}'' is increased by 63.1 and 60.2 $\%$, compared to that of ``{\sf{RM-4color}}'' under perfect CSIT and CSIR and imperfect CSIT and CSIR, respectively.}

{ The traffic demand satisfaction per CSIT accuracy, for a target traffic demand of $\mathbf{r}_{\sf{target}} = [4, 5.6, 1.4, 1.8, 1.8]^{\sf{T}}$, is depicted in Fig. \ref{result_12}. Similar to the trends observed in Fig. \ref{result_6}, the proposed ``{\sf{RM-RSMA}}'' consistently achieves higher traffic demand satisfaction compared to all benchmark schemes across the entire range of $\delta_{\sf{fb}}$ values. Furthermore, as $\delta_{\sf{fb}}$ increases, the performance gap between the proposed scheme and ``{\sf{RM-RSMA (no info $\delta_{\sf{fb}}$, $\delta_{\sf{ce}}$)}}'' becomes more significant. Also, even without leveraging statistical information on phase perturbations, the ``{\sf{RM-RSMA (no info $\delta_{\sf{fb}}$, $\delta_{\sf{ce}}$)}}'' scheme still outperforms other benchmarks across all $\delta_{\sf{fb}}$ values.}


Fig. \ref{result_13} depicts the RM performance of the proposed ``{\sf{RM-RSMA}}'' framework per the regularization parameter $\eta$ of the cases in which the objective functions are set as the L1- and L2-norms, respectively. In other words, the objective functions in $\mathscr{P}_1$ are set as follows:
\begin{align}
\label{L1norm}
  &\min_{\mathbf{P}, \mathbf{c}} \,\, {\sum_{j=1}^{K} \vert R_{{\sf target}{,j}}-R_{j} \vert } = \min_{\mathbf{P}, \mathbf{c}} \,\, {\Vert \mathbf{r}_{{\sf target}}-\mathbf{r} \Vert_{1} },\\
\label{L2norm}
  &\min_{\mathbf{P}, \mathbf{c}} \,\, {\sum_{j=1}^{K} \vert R_{{\sf target}{,j}}-R_{j} \vert ^{2}}  = \min_{\mathbf{P}, \mathbf{c}} \,\, {\Vert \mathbf{r}_{{\sf target}}-\mathbf{r} \Vert_{2}^{2} },
\end{align}
where $\mathbf{r} = [R_{1},\cdots,R_{K}]^{\sf{T}}$. 
The y-axis in Fig. \ref{result_13} represents the average value of the difference between the offered rates and traffic demands. { The larger the value, the more significant the extent to which the offered rates fall short of the demands.

\begin{figure}[!t]
\centering
 		\includegraphics[width=0.78\linewidth]{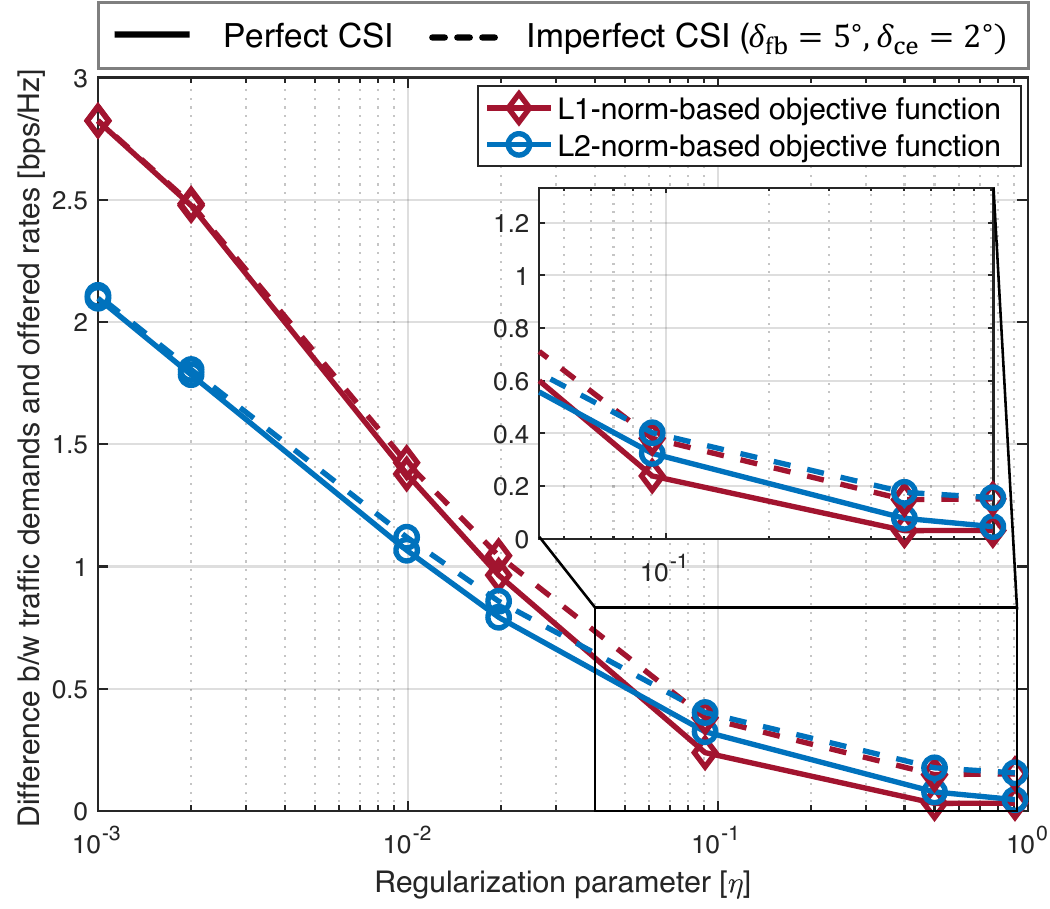}
 		\caption{{Rate-matching performance comparison of the proposed framework when employing the L1-norm and L2-norm-based objective functions when the traffic demand $\mathbf{r}_{\sf{target}}$ is $[2, 2, 3, 3.5, 4]^{\sf{T}}$.}}
    	\label{result_13} 
\end{figure}

\begin{figure}[!t]
\centering
 		\includegraphics[width=0.78\linewidth]{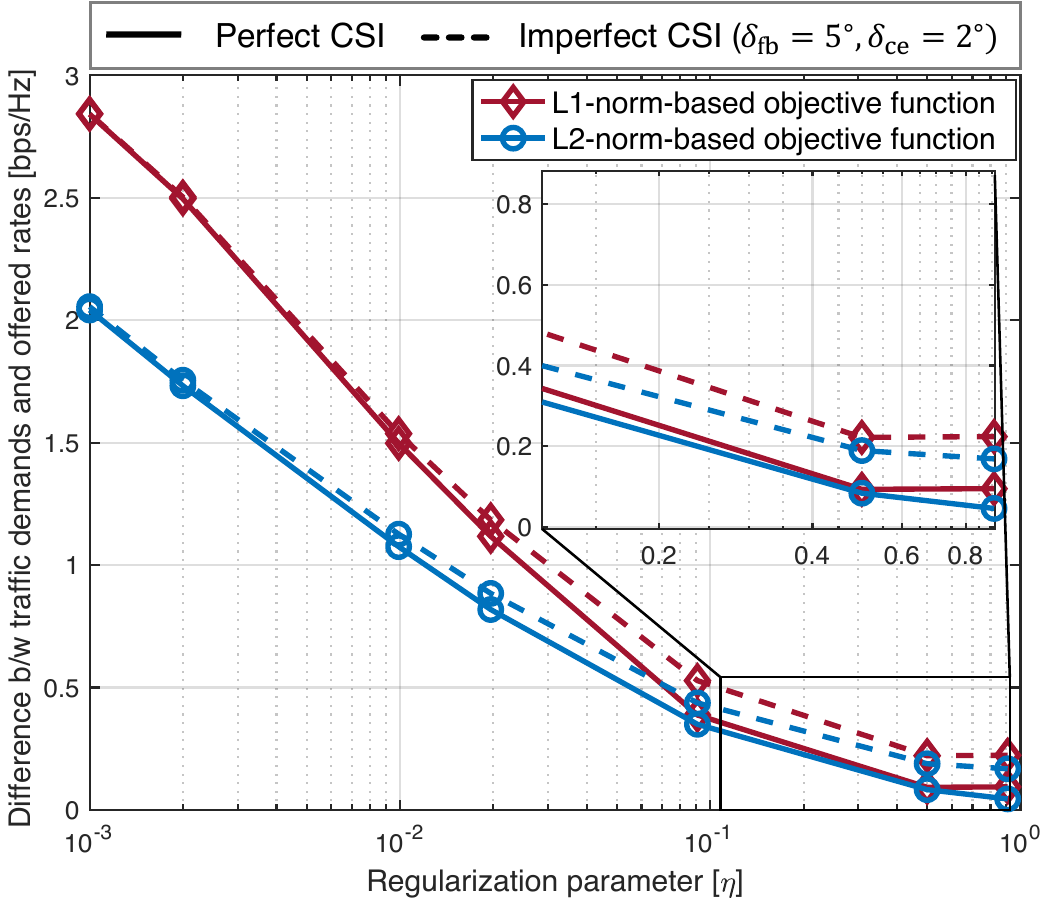}
 		\caption{{Rate-matching performance comparison of the proposed framework when employing the L1-norm and L2-norm-based objective functions when the traffic demand $\mathbf{r}_{\sf{target}}$ is $[4, 5.6, 1.4, 1.8, 1.8]^{\sf{T}}$.}}
    	\label{result_14} 
\end{figure}

We first compare the RM performance of L1-norm and
L2-norm-based objective functions when traffic demand is set to $\mathbf{r}_{\sf{target}}=[2, 2, 3, 3.5, 4]^{\sf{T}}$ bps/Hz.} In general, as $\eta$ increases, the difference between the offered rates and traffic demands decreases since the satellite prioritizes RM rather than transmission power minimization.
The solid line represents the results under the perfect CSIT and CSIR, that is, $\delta_{\sf{fb}} = 0^{\circ}$, $\delta_{\sf{ce}} = 0^{\circ}$, and the dashed line represents the results under the imperfect CSIT and CSIR, that is, $\delta_{\sf{fb}} = 5^{\circ}$, $\delta_{\sf{ce}} = 2^{\circ}$. In both the perfect CSIT and CSIR and the imperfect CSIT and CSIR cases, when $\eta$ is comparatively small, the L2-norm-based objective function performs better than the L1-norm-based objective function in terms of RM.
On the other hand, when $\eta$ becomes comparatively large, the L1-norm-based objective function shows better fulfillment of the traffic demands than the L2-norm-based objective function.
{This is because the L1-norm gives a big weight on small residuals and less weight on large residuals, whereas the L2-norm gives a big weight on large residuals and less weight on small residuals \cite{boyd2004convex}.}

{ We then compare the RM performance of the L1-norm and L2-norm-based objective functions when the traffic demand is set to $\mathbf{r}_{\sf{target}} = [4, 5.6, 1.4, 1.8, 1.8]^{\sf{T}}$ bps/Hz, through Fig. \ref{result_14}. The results show that in scenarios, where traffic demands among users are more uneven, the L2-norm-based framework outperforms the L1-norm-based one across all $\eta$ values, in both the perfect CSIT and CSIR and the imperfect CSIT and CSIR cases. This is because the L2-norm better accommodates disproportionately high traffic demands compared to the L1-norm. By squaring each residual, the L2-norm places greater emphasis on larger residuals, making it more sensitive to significantly higher traffic demands. On the other hand, the L1-norm treats residuals linearly, applying equal weight to all residuals regardless of the traffic demands, which makes it less effective at handling disproportionately high demands.}

\section{Conclusion}
{In this paper, an RSMA-based power-efficient RM framework has been studied for multibeam LEO SATCOM under phase perturbations from both channel estimation and feedback errors.} 
A robust precoder design problem to minimize the difference between the traffic demands and offered rates has been formulated. 
To make this non-convex problem more tractable, we have proposed an SCA method-based iterative algorithm that solves it efficiently. 
 Simulation results have shown that the proposed framework can flexibly design the precoder according to traffic demands. 
Also, the inter-/intra-beam interference can be effectively managed by employing rate splitting approach even with the limited spatial dimension available at satellites. 
The two-fold robustness of the proposed framework over phase perturbation and user-overload scenarios has been demonstrated. 
{Owing to the superior demand-matching capability of the proposed framework, compared with that of the various benchmark schemes, we expect our work to be useful in satisfying the heterogeneous user requirements in upcoming multibeam LEO SATCOM systems.} 

{ Future directions include considering satellite availability based on the geometric orbit model for reliable and continuous satellite-to-user links, as discussed in \cite{tang2021computation, yoo2024cache}. This can be accomplished through coordination between multiple LEO satellites, supported by the exchange of essential information such as user data, ephemeris data, and handover tables.}


\section{Appendix: Proof of (\ref{Derv_4}) and (\ref{Derv_5})}

By using ${\mathbf{e}_{k}^{\sf{fb}}} = [e^{j\theta_{k,1}^{\sf{fb}}}, \cdots, e^{j\theta_{k, N_{{\sf{t}}}}^{\sf{fb}}}]^{\sf{T}}$, $\mathbb{E}[{\mathbf{e}_{k}^{\sf{fb}}}{\mathbf{e}_{k}^{\sf{fb}}}^{\sf{H}}]$ can be represented as
\begin{align} \label{APPENB1}
   \mathbb{E}[{\mathbf{e}_{k}^{\sf{fb}}}{\mathbf{e}_{k}^{\sf{fb}}}^{\sf{H}}] = \mathbb{E}\left[
   \left[\begin{smallmatrix}
   1&\cdots&e^{j(\theta_{k, 1}^{\sf{fb}}-\theta_{k, N_{{\sf{t}}}}^{\sf{fb}})} \\ 
   \vdots&\ddots&\vdots\\
   e^{j(\theta_{k, N_{{\sf{t}}}}^{\sf{fb}}-\theta_{k, 1}^{\sf{fb}})}&\cdots&1
   \end{smallmatrix}\right]
   \right].
\end{align}
As each phase element of ${\mathbf{e}_{k}^{\sf{fb}}}$ follows i.i.d. such that $\mathbf{\theta}_{k}^{\sf{fb}}\sim\mathcal{N}{(0,\delta_{\sf{fb}}^{2} \mathbf{I}_{N_{\sf{t}}})}$, the expectation of the $mn$-th element $\mathbb{E}[e^{j(\theta_{k, m}^{\sf{fb}}-\theta_{k, n}^{\sf{fb}})}]$ is decomposed as $\mathbb{E}[e^{j(\theta_{k, m}^{\sf{fb}})}]\mathbb{E}[e^{-j(\theta_{k, n}^{\sf{fb}})}]$. Herein, $\mathbb{E}[e^{j(\theta_{k, m}^{\sf{fb}})}]$ is rewritten as follows:
\begin{align} \label{APPENB2}
  \mathbb{E}[e^{j(\theta_{k, m}^{\sf{fb}})}] &= \int_{-\infty}^{\infty}{\frac{e^{j(\theta_{k, m}^{\sf{fb}})}}{\sqrt{2\pi\delta_{\sf{fb}}^{2}}}}{e^{-\frac{(\theta_{k, m}^{\sf{fb}})^2}{2\delta_{\sf{fb}}^{2}}}}{d\theta_{k, m}^{\sf{fb}}}  \\ 
  &= \int_{-\infty}^{\infty}{\frac{1}{\sqrt{2\pi\delta_{\sf{fb}}^{2}}}}{e^{-\frac{((\theta_{k, m}^{\sf{fb}})^2-2j\delta_{\sf{fb}}^{2}\theta_{k, m}^{\sf{fb}})}{2\delta_{\sf{fb}}^{2}}}}{d\theta_{k, m}^{\sf{fb}}} \nonumber \\
  &= \int_{-\infty}^{\infty}{\frac{1}{\sqrt{2\pi\delta_{\sf{fb}}^{2}}}}{e^{-\frac{(\theta_{k, m}^{\sf{fb}}-j\delta_{\sf{fb}}^{2})^2}{2\delta_{\sf{fb}}^{2}}}}{e^{-\frac{\delta_{\sf{fb}}^{2}}{2}}}{d\theta_{k, m}^{\sf{fb}}} = {e^{-\frac{\delta_{\sf{fb}}^{2}}{2}}},\nonumber
\end{align}
and $\mathbb{E}[e^{-j(\theta_{k, n}^{\sf{fb}})}]$ can also be rewritten as $\mathbb{E}[e^{-j(\theta_{k, n}^{\sf{fb}})}]=e^{-\frac{\delta_{\sf{fb}}^{2}}{2}}$ through the same procedure. 
Therefore, $\mathbb{E}[e^{j(\theta_{k, m}^{\sf{fb}}-\theta_{k, n}^{\sf{fb}})}]$ becomes $\mathbb{E}[e^{j(\theta_{k, m}^{\sf{fb}}-\theta_{k, n}^{\sf{fb}})}]=\mathbb{E}[e^{j(\theta_{k, m}^{\sf{fb}})}]\mathbb{E}[e^{-j(\theta_{k, n}^{\sf{fb}})}] = e^{-\delta_{\sf{fb}}^{2}}$. In this manner, all non-diagonal elements of (\ref{APPENB1}) can be derived as $e^{-\delta_{\sf{fb}}^{2}}$, whereas all diagonal elements become $1$. 
In conclusion, (\ref{APPENB1}) can be rewritten as follows:
\begin{align} \label{APPENB3}
   \mathbb{E}[{\mathbf{e}_{k}^{\sf{fb}}}{\mathbf{e}_{k}^{\sf{fb}}}^{\sf{H}}] = 
   \begin{bmatrix}
   1&\cdots&e^{-\delta_{\sf{fb}}^{2}} \\ 
   \vdots&\ddots&\vdots\\
   e^{-\delta_{\sf{fb}}^{2}}&\cdots&1
   \end{bmatrix},
\end{align}
and it can be decomposed as follows:
\begin{align} \label{APPENB4}
   \begin{bmatrix}
   1&\cdots&e^{-\delta_{\sf{fb}}^{2}} \\ 
   \vdots&\ddots&\vdots\\
   e^{-\delta_{\sf{fb}}^{2}}&\cdots&1
   \end{bmatrix} 
   = e^{-\delta_{\sf{fb}}^{2}}\mathbf{1}_{N_{{\sf{t}}}} + (1-e^{-\delta_{\sf{fb}}^{2}})\mathbf{I}_{N_{{\sf{t}}}},
\end{align}
which is a PSD matrix because $e^{-\delta_{\sf{fb}}^{2}} \leq 1$. 

Subsequently, using ${\mathbf{e}_{k}^{\sf{ce}}} = [e^{j\theta_{k,1}^{\sf{ce}}}, \cdots, e^{j\theta_{k, N_{{\sf{t}}}}^{\sf{ce}}}]^{\sf{T}}$, $\mathbb{E}[({\mathbf{e}_{k}^{\sf{fb}}}-\mathbf{1})({\mathbf{e}_{k}^{\sf{fb}}}-\mathbf{1})^{\sf{H}}]$ can be represented as follows:
\begin{align} \label{APPENB5}
   &\mathbb{E}[({\mathbf{e}_{k}^{\sf{ce}}} - \mathbf{1})({\mathbf{e}_{k}^{\sf{ce}}} - \mathbf{1})^{\sf{H}}] =
    \\ 
   & \mathbb{E}\left[ \left[
   \begin{smallmatrix}
    2-e^{j(\theta_{k, 1}^{\sf{ce}})}-e^{-j(\theta_{k, 1}^{\sf{ce}})} &  \cdots &   (e^{j(\theta_{k, 1}^{\sf{ce}}-\theta_{k, N_{{\sf{t}}}}^{\sf{ce}})}-e^{j(\theta_{k, 1}^{\sf{ce}})} \\
   && -e^{-j(\theta_{k, N_{{\sf{t}}}}^{\sf{ce}})} + 1) \\ 
   \vdots&\ddots&\vdots\\
   (e^{j(\theta_{k, N_{{\sf{t}}}}^{\sf{ce}}-\theta_{k, 1}^{\sf{ce}})}-e^{j(\theta_{k, N_{{\sf{t}}}}^{\sf{ce}})} && \\
   -e^{-j(\theta_{k, 1}^{\sf{ce}})} + 1) &  \cdots &  2-e^{j(\theta_{k, N_{{\sf{t}}}}^{\sf{ce}})}-e^{-j(\theta_{k, N_{{\sf{t}}}}^{\sf{ce}})}
   \end{smallmatrix}
   \right] \right].\nonumber
\end{align}
Herein, each phase element of $\mathbf{e}_k^{\sf{ce}}$ also follows i.i.d such that $\mathbf{\theta}_{k}^{\sf{ce}}\sim\mathcal{N}{(0,\delta_{\sf{ce}}^{2} \mathbf{I}_{N_{\sf{t}}})}$.
Therefore, through the same procedure as that in (\ref{APPENB2}), all diagonal elements of (\ref{APPENB5}) can be derived as $2-2e^{-\frac{\delta_{\sf{ce}}^{2}}{2}}$, and all non-diagonal elements of (\ref{APPENB5}) can be derived as $e^{-\delta_{\sf{ce}}^{2}}-2e^{-\frac{\delta_{\sf{ce}}^{2}}{2}}+1$. 
Thus, (\ref{APPENB5}) can be rewritten as 
\begin{align} \label{APPENB6}
   &\mathbb{E}[({\mathbf{e}_{k}^{\sf{ce}}}-\mathbf{1})({\mathbf{e}_{k}^{\sf{ce}}}-\mathbf{1})^{\sf{H}}] =
   \nonumber \\
   &\left[ \begin{smallmatrix}
   2-2e^{-\frac{\delta_{\sf{ce}}^{2}}{2}}&\cdots& e^{-\delta_{\sf{ce}}^{2}}-2e^{-\frac{\delta_{\sf{ce}}^{2}}{2}}+1 \\ \vdots&\ddots&\vdots\\
   e^{-\delta_{\sf{ce}}^{2}}-2e^{-\frac{\delta_{\sf{ce}}^{2}}{2}}+1&\cdots&2-2e^{-\frac{\delta_{\sf{ce}}^{2}}{2}}
   \end{smallmatrix} \right] = \nonumber \\
   &\left[ \begin{smallmatrix}
   (1-e^{-\frac{\delta_{\sf{ce}}^{2}}{2}})^2 + 1 - e^{-\delta_{\sf{ce}}^{2}}&\cdots&(1-e^{-\frac{\delta_{\sf{ce}}^{2}}{2}})^2 \\ 
   \vdots&\ddots&\vdots\\
   (1-e^{-\frac{\delta_{\sf{ce}}^{2}}{2}})^2&\cdots&(1-e^{-\frac{\delta_{\sf{ce}}^{2}}{2}})^2 + 1 - e^{-\delta_{\sf{ce}}^{2}}
   \end{smallmatrix} \right],
\end{align}
and it can be decomposed as follows:
\begin{align} \label{APPENB7}
   &\left[ \begin{smallmatrix}
   (1-e^{-\frac{\delta_{\sf{ce}}^{2}}{2}})^2 + 1 - e^{-\delta_{\sf{ce}}^{2}}&\cdots&(1-e^{-\frac{\delta_{\sf{ce}}^{2}}{2}})^2 \\ 
   \vdots&\ddots&\vdots\\
   (1-e^{-\frac{\delta_{\sf{ce}}^{2}}{2}})^2&\cdots&(1-e^{-\frac{\delta_{\sf{ce}}^{2}}{2}})^2 + 1 - e^{-\delta_{\sf{ce}}^{2}}
   \end{smallmatrix} \right]
   \nonumber \\
   &= (1-e^{-\frac{\delta_{\sf{ce}}^{2}}{2}})^2\mathbf{1}_{N_{\sf{t}}} + (1 - e^{-\delta_{\sf{ce}}^{2}}) \mathbf{I}_{N_{\sf{t}}},
\end{align}
which is also a PSD matrix, because $e^{-\frac{\delta_{\sf{ce}}^{2}}{2}} \leq 1$ and $e^{-\delta_{\sf{ce}}^{2}} \leq 1$. 
Further, as the Hadamard product of the PSD matrices is a PSD matrix, $\hat{\mathbf{h}}_{k} \hat{\mathbf{h}}_{k}^{\sf{H}} \odot
    [e^{-\delta_{\sf{fb}}^{2}}\mathbf{1}_{N_{{\sf{t}}}} + (1-e^{-\delta_{\sf{fb}}^{2}})\mathbf{I}_{N_{{\sf{t}}}}]$ and $\hat{\mathbf{h}}_{k} \hat{\mathbf{h}}_{k}^{\sf{H}} \odot
    [(1-e^{-\frac{\delta_{\sf{ce}}^{2}}{2}})^{2}\mathbf{1}_{N_{{\sf{t}}}} + (1-e^{-\delta_{\sf{ce}}^{2}})\mathbf{I}_{N_{{\sf{t}}}}]$ are also PSD matrices.

\bibliographystyle{IEEEtran}
\bibliography{jhseong_reff}
\end{document}